\documentclass[10 pt, journal, compsoc]{IEEEtran}  
\ifCLASSOPTIONcompsoc
  \usepackage[nocompress]{cite}
\else
  \usepackage{cite}
\fi                          

\usepackage{fancyhdr}
\usepackage{amsmath}
\usepackage{parskip}
\usepackage{listings}
\usepackage{indentfirst}
\usepackage{graphicx}
\usepackage{float}
\usepackage{extarrows}
\usepackage{amssymb}
\usepackage{amsthm}
\usepackage{subfigure}
\usepackage{epstopdf}
\usepackage{color}

\ifCLASSINFOpdf
\else
\fi

\begin{document}
\title{\LARGE \bf
The Interplay of Competition and Cooperation Among Service Providers (Part \uppercase\expandafter{\romannumeral2})
}

\author{Xingran Chen,\IEEEmembership{}
        Saswati Sarkar\IEEEmembership{}, Mohammad Hassan Lotfi,\IEEEmembership{}
\IEEEcompsocitemizethanks {\IEEEcompsocthanksitem Xingran Chen is with the Electrical and System Engineering Department of the University of Pennsylvania, PA, 19104. \protect\\
E-mail: xingranc@seas.upenn.edu
\IEEEcompsocthanksitem Saswati Sarkar is with the Electrical and System Engineering Department of the University of Pennsylvania, PA, 19104. \protect\\
E-mail: swati@seas.upenn.edu
\IEEEcompsocthanksitem  Mohammad Hassan Lotfi is  at Capital One.\protect\\
E-mail: mohammadhassanlotfi@gmail.com
}
}

\newtheorem{lemma}{Lemma}
\newtheorem{note}{Note}
\newtheorem{property}{Property}
\newtheorem{theorem}{Theorem}
\newtheorem{definition}{Definition}
\newtheorem{corollary}{Corollary}
\newtheorem{proposition}{Proposition}
\newtheorem{remark}{Remark}
\newtheorem{assumption}{Assumption}

\IEEEtitleabstractindextext{
\begin{abstract}
 This paper investigates the incentives of mobile network operators (MNOs) for acquiring
additional  spectrum to offer mobile virtual network operators (MVNOs) and thereby inviting competition
for a common pool of end users (EUs). We consider interactions between two service providers (SPs), a  MNO and an MVNO, under two different scenarios:   1) EUs must choose one of them 2) EUs have the option to defect to a provider outside the system under consideration should the SP duo  offer unsatisfactory access fees or qualities of service. We formulate a multi-stage hybrid of  cooperative bargaining  and  non-cooperative games. In this formulation, first the  two SPs jointly determine their spectrum acquisitions, allocations and mutual money flows through the  bargaining game;  subsequently the two SPs individually determine the access fees for the EUs through the non-cooperative game. We identify when the overall equilibrium solutions exist and  when the equilibrium solution is unique. We obtain computationally simple characterizations of  the  equilibrium solutions when they exist, which are in closed form or involve optimizations in only one decision variable. The hybrid framework allows us to determine whether and by how much the different entities benefit due to the cooperation in spectrum acquisition decision.

\end{abstract}


}

\maketitle
\IEEEdisplaynontitleabstractindextext
\IEEEpeerreviewmaketitle

\IEEEraisesectionheading{\section{Introduction}\label{sec: introduction}}

\subsection{Motivation and Overview}
\IEEEPARstart{T}{wo} different classes of service providers co-exist in the current wireless service provider (SP) market: Mobile Networks Operators (MNOs) and Mobile Virtual Network Operators (MVNOs).  The MNO acquires $I_L$ spectrum from a regulator, which he offers  to a MVNO in exchange of money, and the MVNO  uses $I_F$ amount of this spectrum. Both   SPs earn by selling  wireless plans to  end users (EUs); the MNO earns additionally  by leasing her spectrum to the MVNO.   Thus, they both cooperate, by sharing spectrum; they also compete, for a common pool of EUs. They clearly make different decisions, which affect their subscriptions; their payoffs have different expressions and their  decisions also follow different constraints, eg,  $I_L$ can be chosen to be any positive value\footnote{ Specifically, $I_L \geq \delta > 0$ where lower bound $\delta$ is a parameter of choice.  Since parameter $\delta$ may be chosen as low a positive quantity as one desires, for all practical purposes, $I_L$ can be chosen to be any positive value. The assumption $I_L \geq \delta > 0$, rather than $I_L \geq 0$ or $I_L > 0$ simplifies our analysis, and is not restrictive as argued above.}     and $I_F$ must be chosen as a value between $0$ and $I_L$  (i.e., $0 \leq I_F \leq I_L$). In a sequence of two papers we investigate the economics of the interplay of the competition and cooperation between an  MNO and an MVNO.

\subsection{Relation with the Prequel}

In the prequel, Part~\uppercase\expandafter{\romannumeral1}~\cite{Part1},  we consider that the SPs arrive at their decisions individually,  in the current paper  we consider that the SPs arrive at certain   decisions together, and then arrive at other  decisions individually. Specifically, in this paper,   the SPs together decide  the spectrum they  acquire (i.e., $I_L, I_F$) to maximize their overall profits, and the  marginal reservation fee $s$  that the MVNO pays to the MNO for using the spectrum the MNO offers. Here $s$ is decided so as to split the proceeds between the SPs in accordance with the subscription revenue each generates, which in turn depends on the prior preferences of the EUs for them.   Subsequently, each SP individually  decides the  access fees for the EUs. The $I_L, I_F, s$ are obtained as  the solution of a cooperative bargaining game, and the access fees are obtained as solutions of a non-cooperative game. The bargaining and the non-cooperative games together constitute a sequential game. In contrast, in the prequel, each decision variable is selected through a non-cooperative game, each of which constitutes a stage of a sequential game. Also, the marginal reservation fee is considered a fixed parameter, and the MNO and MVNO individually decides the spectrum each acquires, and subsequently individually decides the access fees for the EUs. Note that    the marginal reservation fee is indeed a market-driven parameter  in a large spectrum market with many MNOs and MVNOs; in such a scenario the marginal reservation fee may be driven by the overall market evolution, and is beyond the control of individual MNOs and MVNOs. This fee may also be beyond the control of individual MNOs and MVNOs, when it is determined by an external regulator to influence the interaction between different providers perhaps to protect the interests of the EUs. These are the cases that the prequel considers.  In a smaller market and in absence of regulatory intervention,  the marginal reservation fee would be chosen as a decision variable through a negotiation between the MNO and the MVNO concerned. This is the case this paper considers.

\subsection{Positioning vis-a-vis the State-of-the-Art}

The economics of the interactions of resource sharing among service providers have been investigated in many works. In the prequel, we have distinguished our contributions from those in the genre of non-cooperative interaction between the SPs, since there we considered that they arrive at their decisions individually. In this paper, since the SPs together decide the spectrum they acquire, we review the state of the art on cooperative interaction between providers, which have invariably been modeled by coalitional and bargaining games.

Coalitional games were investigated in \cite{S2012}, \cite{S2012-2} and \cite{CCCCP2016}. Transferable and nontransferable payoff coalitional games were used in \cite{S2012} and \cite{S2012-2}, respectively, to model cooperation among service providers through joint deploying and pooling of resources and  serving each others' customers.  Both papers concluded that cooperation substantially enhances individual providers' payoffs. In \cite{CCCCP2016}, MNOs weighed between building individual networks or entering into network and spectrum sharing agreements. Coalitional games with transferable and nontransferable utility were built to show that a cost division policy guaranteed coalition stability.

Bargaining games were studied in
\cite{GZ2009}, \cite{JS2009}, \cite{GZ2011}, \cite{LG2011}, \cite{YHW2013}, \cite{ZJB2015}, and \cite{BC2015}.
The cooperation between selfish nodes was formulated as two-person bargaining games in \cite{GZ2009}, \cite{GZ2011}, both nodes were seen to perform better than if they work independently. In \cite{JS2009}, nodes in a wireless network seek to agree on a fair and efficient allocation of spectrum. Nash Bargaining Solution (NBS) achieves the best tradeoff between fairness and efficiency. 
 A dynamic incomplete information bargaining was built in \cite{YHW2013}, where the primary user does not have complete information of the second user energy cost. NBS can lead to a win-win situation,  i.e., data rate of both users are improved. \cite{ZJB2015} investigated the joint uplink sub-channel and power allocation problem in cognitive small cells with imperfect channel state information. \cite{BC2015} modeled a situation of dynamic spectrum access by a set of cognitive radio enabled nodes as a bargaining game where the nodes bargain among themselves in a distributed manner to agree upon a sharing rule of the channels. The selfish strategies of the players affect system wide performance.
Other optimization models were introduced in \cite{V2012,KMZR2012,TL2013,KSY2014}, and fuzzy logic based frameworks was considered in \cite{CSKBA2015}.

However, these works do not consider the dynamics of the interplay of competition and cooperation  between MNOs and MVNOs, whose roles are fundamentally different from each other. The principal difference is that  while both   MNO and MVNO earn by selling  wireless plans to the EUs,  the MNO earns additionally  by leasing  spectrum to the MVNO. Thus, they make different decisions, which affect their subscriptions, and  their payoffs have different expressions. To our knowledge, \cite{LG2011} is the only work in the domain of cooperative interaction between SPs,  that also considers the dynamics of providers whose roles are similar to those of the MNO and MVNO. This paper considers that the spectrum the MNO acquires is exogenously determined, whereas we consider this as a joint decision of the 2 SPs. This leads to an additional stage in our multi-stage formulation. The subscription models for the EUs are also different, though in both cases the EUs choose between the SPs based on the access fees and the spectrum availability (quality of service in  \cite{LG2011}). Though our model is more general in that it consists of an additional decision variable and additional stage, we are able to obtain the closed form expressions for SPNE 1) access fees, 2) the amount of spectrum the MVNO leases from the MNO, and 3) the reservation fee the MVNO pays to the MNO. In contrast, \cite{LG2011} only proves that the SPNE access fees exist,   and provides the feasibility region of 2) and 3). We also obtain closed form expressions for SPNE spectrum acquisition of the MNO from the central regulator, which \cite{LG2011} considers as a given parameter.  We also generalize our model and results to allow for the possibility that the EUs do not choose either the MNO or the MVNO, but chooses some other SP outside the system we consider; \cite{LG2011} does not do this generalization.

%
 The only other papers that consider the dynamics of MNO and MVNO, namely   \cite{Banerjee2009}, \cite{le2009pricing} and \cite{C2012}, have considered only non-cooperative decisions by the SPs. We have therefore distinguished these from our contributions in the prequel, which is closer to them.


\subsection{Contribution}

We now describe the contributions of this paper. First, we consider a base case in which one MNO and one MVNO compete for EUs in a common pool, and the EUs  choose one of the SPs through a hoteling model for subscription  (Section~\ref{sec: cooperative base case}). We formulate the sequential hybrid of bargaining and non-cooperative games that model the dynamics of the SP interactions (Section~\ref{sec: systemmodel1}), and identify the salient properties of its equilibrium solutions when they exist (Section~\ref{standard}). We obtain conditions for existence and uniqueness of the equilibrium solutions in terms of system parameters, and  characterize them when they exist (Section~\ref{sec: Un-outcome}). We prove that the bargaining framework yields a collusive outcome in which the MNO acquires the minimum amount of spectrum that he is mandated to and the MVNO leases either all or nothing of this spectrum from the MNO (though the MVNO is allowed to lease any amount of this spectrum). The equilibrium solutions are easy to compute and reveal several underlying insights: eg, only the SP that is apriori more popular retains the spectrum leased from the regulator in its entirety. This spectrum sharing arrangement is obtained strategically to motivate the EUs to choose the SP that offers higher price so that the overall subscription revenue is maximized (since the proceeds are shared between the SPs anyway). 
Comparing the payoffs of the SPs and the access fees for the EUs in this paper with those obtained in Part~\uppercase\expandafter{\romannumeral1}~\cite{Part1}, we show that joint decision on spectrum acquisition conclusively benefits the SPs by considerably enhancing their payoffs. The joint decision provides only nuanced benefits for the EUs,   by securing cheaper access fees for them, while simultaneously guiding more EUs to more expensive service by having the more apriori popular SP retain the acquired spectrum in its entirety, and
thereby provide better quality of service to the EUs.  Accordingly, as compared to individual decisions,  for some parameter values the EU-resource-cost metric that we define in Part~\uppercase\expandafter{\romannumeral1}~\cite{Part1} is higher under the joint decision,   and lower for the rest (Section~\ref{numerical}).

  Next, we  allow the EUs to choose a SP outside the system we consider,  if neither of the two SPs in the system (that is, the MVNO and the MNO)  offer a desirable combination of access fee and quality of service.  We also allow each SP in the system to have exclusive additional customer bases to draw from depending on his spectrum acquisition and the price he offers (Section \ref{sec: outside option}).   In this scenario we show that there are two equilibrium solutions, both of which yield a milder version of the  collusive outcome than  in the base case, in that the MNO may acquire higher than the mandated minimum amount of spectrum (Sections~\ref{analysis}, \ref{numericaloutside}). This happens because the EUs have an outside option to desert to, and the SPs have exclusive customer bases to gain from, depending on the price and the qualities of service they offer.    The two equilibrium solutions differ in which of the SPs retain the spectrum leased from  the  regulator. The SP that retains the entire spectrum gets a higher payoff in each case. Under both equilibrium-type solutions, each SP increases his payoff compared to what he gets when the SPs decide their spectrum acquisitions individually.  Also, the EU-resource-cost metric is invariably higher than when the SPs decide their spectrum acquisitions individually.

\section{Base Case}\label{sec: cooperative base case}
We formulate the dynamics of interaction between the SPs as a sequential hybrid of bargaining and non-cooperative games in Section~\ref{sec: systemmodel1}, we identify some salient properties of its equilibrium-type solutions in Section~\ref{standard} and characterize the equilibrium-type solutions in Section~\ref{sec: Un-outcome}. Using these solutions, we assess how the SPs and the EUs fare due to the cooperation between the SPs in jointly deciding their spectrum acquisitions, compared to when they decide everything individually, through analysis in Sections~\ref{sec: Un-outcome} and  through numerical computations in Section~\ref{numerical}.
\subsection{Model}\label{sec: systemmodel1}
We start with by recapitulating  notations that are similar to Part~\uppercase\expandafter{\romannumeral1}~\cite{Part1} and the current paper.  We denote MNO as $\text{SP}_{L}$ and MVNO as $\text{SP}_{F}$.  $\text{SP}_{L}$   offers  $I_L$ amount of  spectrum  (which it acquires from a central regulator) to $\text{SP}_{F}$  in exchange of money, and $\text{SP}_{F}$ uses $I_F$ amount of this spectrum. Clearly, $0\leq I_{F}\leq I_{L}$. 
 We denote the marginal leasing fee (per spectrum unit)  that $\text{SP}_{L}$ pays the central regulator as $\gamma$, marginal reservation fee   $\text{SP}_{F}$ pays to $\text{SP}_{L}$ by $\tilde{s}$, an additional remuneration that $\text{SP}_{L}$ transfers to $\text{SP}_{F}$ by $\theta$,   the fraction of EUs that $\text{SP}_{F}$ and $\text{SP}_{L}$ attract as $n_{F}$ and $n_{L}$, respectively,  and  the access fee that $\text{SP}_{F}$ and $\text{SP}_{L}$ charge the EUs as $p_{F}$ and $p_{L}$, respectively. Let $c$ be the transaction cost incurred by a SP for each subscription. 
The SP$_{L}$ incurs a spectrum acquisition cost of  $\gamma I_{L}^{2}$, and SP$_F$ pays to SP$_L$ a leasing fee of  $sI_{F}^{2}$.
Thus, SP$_{L}$, SP$_{F}$ receive payoffs $\pi_F, \pi_L$ respectively, where:
\begin{align}
 &\pi_{F}=n_{F}(p_{F}-c)-\tilde{s}I_{F}^{2} + \theta \label{equ: BM-L-payoff}\\
 &\pi_{L}=n_{L}(p_{L}-c)+\tilde{s}I_{F}^{2}-\gamma I_{L}^{2} - \theta.\label{equ: BM-F-payoff}
\end{align}
The above equations are similar to (1), (2) of Part~\uppercase\expandafter{\romannumeral1}~\cite{Part1}, with the exception of the introduction of  $\theta$ whose significance will be explained later.

We use a hotelling model to describe how EUs choose between the SPs. EUs are distributed uniformly along the unit interval $[0,1]$, and SP$_{L}$ and SP$_{F}$ are respectively located at $0, 1$  (Figure $1$ of Part~\uppercase\expandafter{\romannumeral1}~\cite{Part1}).
Let $t_{L}$ ($t_{F}$) be the unit transport cost of EUs for SP$_{L}$ (SP$_{F}$), the EU located at $x\in[0,1]$ incurs a cost of $t_{L}x$ (respectively, $t_{F}(1-x)$) when joining SP$_{L}$ (respectively, SP$_{F}$).
The transport costs capture the impact of the qualities of services the SPs offer on the subscription of the EUs, which in turn depend on the spectrum they acquire: $t_L = I_F/I_L, t_F = 1-t_L.$
$v_L, v_F$ represent prior preferences of the EUs for SP$_L$, SP$_F$ respectively, which is the same for all EUs, and do not depend on the strategies of the SPs , i.e., $I_L, I_F, p_L, p_F$.  The EU at $x$ receives utilities $u_{L}(x), u_{F}(x)$ respectively  from SP$_{L}$ and  SP$_{F}$, and joins the SP that gives it the higher utility, where:
\begin{equation}\label{equ: BM-utility EUs}
\begin{aligned}
u_{L}(x)&=v^{L}-\left(p_{L}+t_{L}x\right) \\
u_{F}(x)&=v^{F}-\left(p_{F}+t_{F}(1-x)\right).
\end{aligned}	
\end{equation}
As in Part I \cite{Part1}, we denote $\Delta = v^L - v^F.$

We now mention the major differences with Part~\uppercase\expandafter{\romannumeral1}~\cite{Part1}. Here,
we consider a hybrid of bargaining and non-cooperative games to model the dynamics of the interaction between SP$_L$ and SP$_F$.
The two SPs jointly decide on the spectrum acquisitions ($I_L$, $I_F$), so as to maximize the overall profit,  but individually decide on the access fees for EUs, $p_L, p_F$.  The SPs also split the profit, by selecting the marginal reservation fee $\tilde{s}$, and the additional remuneration $\theta$. 
Thus, $\tilde{s}, \theta$ are new decision variables\footnote{A question that arises is if  the SPs  jointly decide the spectrum acquisitions, why would they not jointly select the access fees too. The answer is two-fold. First,  SP$_L$ offers the spectrum he acquires to SP$_F$, a part of which SP$_F$ uses - thus, they share the spectrum anyhow, that is, the spectrum usage is inherently cooperative. On the other hand, they are competing for the same pool of EUs, it is therefore natural that the access fees will be determined competitively, thus such decisions must be individual. Second, in practice, the spectrums are acquired for larger time intervals, while access fees are updated more frequently. Joint decisions between two SPs involves substantial coordination and negotiation, which is infeasible on shorter time scales.}. The SPs decide $I_L, I_F, \tilde{s}, \theta$ through a bargaining process. If the SPs, SP$_{L}$ SP$_{F}$, are not able to agree on these,  they receive their respective  \emph{disagreement payoff}s, $d_L, d_F$, which we assume to be equal to their payoffs in the sequential non-cooperative game  whose outcome was characterized in Part~\uppercase\expandafter{\romannumeral1}~\cite{Part1} (Theorems 1, 2). 
The disagreement payoff is for example higher for a SP who is apriori more popular, i.e., has a larger $v_L$ or $v_F$, (eg, Figure 4 of Part~\uppercase\expandafter{\romannumeral1}~\cite{Part1}. The disagreement payoffs also depend on the marginal fee per spectrum unit $s$ the SP$_F$ pays the SP$_L$ in the event of a disagreement. This marginal fee is a parameter determined by the overall spectrum market, as  assumed for $s$ in Part~\uppercase\expandafter{\romannumeral1}~\cite{Part1}. 
We also define a \emph{bargaining power} of the SPs. Let $0\leq w\leq 1$ be the relative bargaining power of the SP$_F$ over SP$_L$: the higher the $w$, more is SP$_F$'s  bargaining power.

In the event of agreement, the SPs decide their shares of the overall profit, and thereby $\tilde{s}, \theta$, commensurate with their disagreement payoffs and bargaining powers; higher values of the latter two fetch higher shares of the profit. Since $\tilde{s}$ will have no significance in deciding the shares if $I_F$ is decided as  $0$ \big(refer to \eqref{equ: BM-L-payoff} and \eqref{equ: BM-F-payoff}\big), we have considered the additional remuneration transfer decision variable $\theta$ (which was not in Part~\uppercase\expandafter{\romannumeral1}~\cite{Part1}). Note that $\theta$ can be positive or negative, and the sign reflects the direction of the money flow.

When the SPs jointly decide the spectrum to acquire, so as to maximize the overall profits,  a collusive outcome may occur in which both SPs jointly decrease the amount of spectrum acquisitions while maintaining a specific relative difference that yields the best outcome. The reason is that EUs decide based on the ratio of the investment by SPs and not the absolute values. Thus, regulatory intervention may be desirable. Therefore, we consider that
a regulator  enforces a minimum spectrum acquisition amount of $L_{0}$ on SP$_L$, i.e., $0<L_0\leq I_L$. Recall that we have a minimum required amount for $I_L$, $\delta$, in Part~\uppercase\expandafter{\romannumeral1}~\cite{Part1}, $L_0$ may not be the same as the  $\delta$. This is because collusion does not naturally arise in the non-cooperative selection in Part~\uppercase\expandafter{\romannumeral1}~\cite{Part1}. Thus, a minimum amount $\delta$ was mandated merely for convenience of analysis, and $\delta$ was assumed small everywhere. Here, the minimum amount $L_0$ is imposed as a regulatory intervention to ensure some minimum quality of service for the EUs in presence of collusion between the SPs.

We  formulate a bargaining framework and use  the \emph{Nash Bargaining Solution} (NBS) to characterize $I_F, I_L, \tilde{s}, \theta$:

 \begin{definition}\emph{Nash Bargaining Solution (NBS):} is the unique solution (in our case the tuple of the payoffs of SP$_L$ and SP$_F$) that satisfies the four ``reasonable" axioms (Invariant to affine transformations, Pareto optimality, Independence of irrelevant alternatives, and Symmetry) characterized in \cite{MO1990}.
\end{definition}

From standard game theoretic results in \cite{MO1990}, the optimal solution of the following maximization,  $(\pi_{L}^{*}, \pi_{F}^{*})$, constitute the Nash Bargaining Solution:
\begin{equation}\label{equ: BM-bargaining optimization-1}
\begin{aligned}
\max_{\pi_{L}, \pi_{F}}\quad&(\pi_{F}-d_{F})^{w}(\pi_{L}-d_{L})^{1-w}\\
s.t\quad&(\pi_{L},\pi_{F})\in U, \ \
(\pi_{L},\pi_{F})\geq(d_{L},d_{F})
\end{aligned}
\end{equation}
where
\begin{align*}
U=&\left\{(\pi_{F}, \pi_{L})|\begin{aligned}
&\pi_{F}=n_{F}(p_{F}-c)-\tilde{s}I_{F}^{2}+\theta\\
&\pi_{L}=n_{L}(p_{L}-c)+\tilde{s}I_{F}^{2}-\theta-\gamma I_{L}^{2}	
\end{aligned}\right\}\\
\cap&\left\{ L_0\leq I_{L}, 0\leq I_{F}\leq I_{L} \right\}.
\end{align*}

\begin{remark}
\label{remark4repeat}
Thus, the payoffs of the individual SPs after bargaining is no less than their disagreement payoffs.
\end{remark}

\begin{remark}
\label{remark1}
The above optimization is guaranteed to have a feasible solution if $L_0$ is lower than the spectrum acquisition of SP$_L$ that corresponds to his disagreement payoff; it need not have a feasible solution otherwise.
\end{remark}

 The SPs decide $I_L, I_F, \tilde{s}, \theta$ as per the following sequential hybrid of bargaining and non-cooperative games:
\begin{itemize}
	\item {\bf Stage 1:} $\text{SP}_{L}$ and $\text{SP}_{F}$ jointly decide  $(I_{L}, I_{F}, \tilde{s}, \theta)$ through the  bargaining game \eqref{equ: BM-bargaining optimization-1}.
	\item {\bf Stage 2:} $\text{SP}_{L}$ and $\text{SP}_{F}$ determine the $p_{L}$ and $p_{F}$, respectively, and individually, to maximize their payoffs $\pi_L, \pi_F$, based on $I_L, I_F, \tilde{s}, \theta$ determined in the previous stage. The process constitutes a non-cooperative game.
	\item {\bf Stage 3:} EUs decide to subscribe to one of the SPs based on $I_L, I_F, p_L, p_F$ determined in the previous stages and prior preferences $v^L, v^F$. A EU at location $x$ chooses the SP that provides it a higher utility as per the expressions in \eqref{equ: BM-utility EUs}.
\end{itemize}


 From the above, $n_F, n_L, p_L, p_F$ are determined in Stage $2$ based on  $I_{L}, I_{F}, \tilde{s}, \theta$ determined in Stage $1$, as solution of \eqref{equ: BM-bargaining optimization-1}. Thus, $n_F, n_L, p_L, p_F$ are functions of $I_{L}, I_{F}, \tilde{s}, \theta$; therefore the latter  are the decision variables in optimization \eqref{equ: BM-bargaining optimization-1}. Thus optimization \eqref{equ: BM-bargaining optimization-1} is
\begin{equation}\label{equ: BM-bargainging optimization-2}
\begin{aligned}
\max_{I_{L},I_{F}, \tilde{s}, \theta}\quad&(\pi_{F}-d_{F})^{w}(\pi_{L}-d_{L})^{1-w}\\
s.t\quad&0\leq I_{F}\leq I_{L},\,\,\,L_0\leq I_{L}\\
&\pi_{F}=n_{F}(p_{F}-c)-\tilde{s}I_{F}^{2}+\theta\\
&\pi_{L}=n_{L}(p_{L}-c)+\tilde{s}I_{F}^{2}-\theta-\gamma I_{L}^{2}\\
&(\pi_{L},\pi_{F})\geq(d_{L},d_{F})
\end{aligned}
\end{equation}

\begin{definition}
We define $(I_{L}^{*}, I_{F}^{*}, \tilde{s}^{*}, \theta^*,  p_{L}^{*},p_{F}^{*},n_{L}^{*},n_{F}^{*})$ as an  {\it equilibrium-type solution}, when $I_{L}^{*}, I_{F}^{*}, \tilde{s}^{*}, \theta^*$ constitute the optimum solution of
\eqref{equ: BM-bargainging optimization-2}, $p_L^*, p_F^*$ the Nash equilibrium of the non-cooperative game in Stage~$2$, and $n_L^*, n_F^*$ the corresponding EU subscriptions in Stage~$3$. Let $(\pi_{L}^{*}, \pi_{F}^{*})$ be the corresponding payoffs of the SPs,
\end{definition}

 If an equilibrium-type solution exists, it may be determined through  backward induction,  starting from the last stage (stage 3) of the game and proceeding backward.

\begin{remark}\label{remark0}
There is for example no equilibrium-type solution if \eqref{equ: BM-bargainging optimization-2} does not have a feasible solution.
\end{remark}


Note that the framework presented above is identical to that in Sections 2.1, 2.2 of Part~\uppercase\expandafter{\romannumeral1}~\cite{Part1} except that 1) $I_L^*, I_F^*, \tilde{s}^*, \theta^*$ are determined as solutions of a bargaining game as opposed to $I_L^*, I_F^*$ being obtained as SPNE of a non-cooperative game and 2) $s$ being a fixed parameter in and $\theta$ not being invoked in Part~\uppercase\expandafter{\romannumeral1}~\cite{Part1}. Thus, once we get an optimum $(I_{L}^{*}, I_{F}^{*}, \tilde{s}^*, \theta^*)$, from \eqref{equ: BM-bargainging optimization-2}, the access fee for EUs ($p_L^*$ and $p_F^*$) and the split of EUs ($n_L^*$ and $ n_F^*$)  between SPs can be determined from the results in Part~\uppercase\expandafter{\romannumeral1}~\cite{Part1}, namely Theorems~1, 2, depending on the value of $\Delta$. In fact, Theorems~1, 2 of Part~\uppercase\expandafter{\romannumeral1}~\cite{Part1} show that $p_L^*, p_F^*, n_L^*,  n_F^*$ are expressions only of $I_{L}^{*}, I_{F}^{*}$:
 \begin{theorem}\label{thm1Part1}
 [Theorem 1 of Part~\uppercase\expandafter{\romannumeral1}~\cite{Part1}]
Let $|\Delta| < 1$. The SPNE  $p_L^*, p_F^*, n_L^*,  n_F^*$ are:

\noindent{\bf (1)}
$p_{L}^{*}=c+\frac{2}{3}-\frac{I_{F}^{*}}{3I_{L}^{*}}+\frac{\Delta}{3},\quad
p_{F}^{*}=c+\frac{1}{3}+\frac{I_{F}^{*}}{3I_{L}^{*}}-\frac{\Delta}{3}$,

\noindent{\bf (2)} $n_{L}^{*}=\frac{\Delta}{3}+\frac{2}{3}-\frac{I_{F}^{*}}{3I_{L}^{*}},\, n_{F}^{*}=\frac{I_{F}^{*}}{3I_{L}^{*}}+\frac{1}{3}-\frac{\Delta}{3}$.
\end{theorem}

\begin{theorem}\label{thm2Part1}[Theorem 2 of Part~\uppercase\expandafter{\romannumeral1}~\cite{Part1}]
The  SPNE  $p_L^*, p_F^*, n_L^*,  n_F^*$ are

\noindent {\bf (1) } $\Delta\geq 1$:
\[p_{F}^{*}=p_{L}^{*}-\Delta,
n_{L}^{*}=1,\, n_{F}^{*}=0,\] and
$p_L^*$ can be chosen any value in $[c+1, c+\Delta].$ \newline
 {\bf (2) } $\Delta = 1:$ The following interior strategy constitute an additional  SPNE:
\[p_L^* - c = n_L^* = 2/3, p_F^* - c = n_F^*= 1/3.\]
{\bf (3) } $\Delta <  -1:$
 \[p_{L}^{*}=p_{F}^{*}+\Delta-1, n_{L}^{*}=0,\, n_{F}^{*}=1,\] and
$p_L^*$ can be chosen any value in $[c+1, c-\Delta].$ \newline
\end{theorem}
Using the above, we now proceed to determine $(I_{L}^{*}, I_{F}^{*}, \tilde{s}^*, \theta)$ in the next two sections.
These, together with $\tilde{s}^*, \theta^*$,  will provide the payoffs of the individual SPs, $\pi_L^*, \pi_F^*$.

\subsection{Properties of the equilibrium-type solutions}
\label{standard}
We now obtain identify some properties of the  equilibrium-type solutions.

We define the {\it aggregate excess profit} to be the additional profit yielded from the cooperation in the bargaining framework:
\begin{definition}\label{def: BM-u-excess}
Aggregate Excess Profit ($u_{excess}$): The aggregate excess profit is defined as
\begin{equation}\label{equ: BM-u-excess}
\begin{aligned}
u_{excess}=&\pi_{L}-d_{L}+\pi_{F}-d_{F} \\
= & n_{F}(p_{F}-c)+n_{L}(p_{L}-c)-\gamma I_{L}^2-d_{L}-d_{F}
\end{aligned}
\end{equation}	
\end{definition}


We have argued in the last paragraph of Section~\ref{sec: systemmodel1} that the equilibrium-type $p_L^*, p_F^*, n_L^*,  n_F^*$ are expressions only of $I_{L}^{*}, I_{F}^{*}$. Thus, under the equilibrium-type solutions,  $u_{excess}$ is only a function of $I_{F}^*, I_{L}^*, d_F, d_L$. We denote $u_{excess}^{*}=u_{excess}|_{I_{L}=I_{L}^{*}\&I_{F}=I_{F}^{*}}$.

\begin{theorem}\label{thm: payoffs in bargaining are better}
 The equilibrium-type payoffs of SPs satisfy the following property:
\begin{align}
&\pi_{L}^{*}=(1-w)u_{excess}^*+d_L\label{equ: pi_L-basecase} \\
&\pi_{F}^{*}=wu_{excess}^*+d_F.\label{equ: pi_F-basecase}
\end{align}
\end{theorem}
\begin{remark}
The SPs split $u_{excess}^{*}$ based on their relative bargaining power, SP$_{F}$ obtains a portion $w$, and SP$_{L}$ obtains the rest. Each SP's payoff equals his share of this aggregate excess profit plus his disagreement payoff. Thus, his payoff increases with his bargaining power and his disagreement payoff; the latter depends on  $|\Delta|, s, \gamma.$ 	
\end{remark}

\proof
 From (2) in \cite{AM2002}, the NBS $(\pi_{L}^{*}, \pi_{F}^{*})$  satisfies:
 \begin{align}\label{equ: BM-NBS-fraction}
\frac{\pi_{F}^{*}-d_{F}}{w}=\frac{\pi_{L}^{*}-d_{L}}{1-w}.
\end{align}
\begin{align}\label{equ: BM-pi_L-d_L}
 \mbox{From (\ref{equ: BM-NBS-fraction}), } \ \  \pi_{L}^*-d_{L}=\frac{1-w}{w}(\pi_{F}^*-d_{F}).
\end{align}
Substituting (\ref{equ: BM-pi_L-d_L}) into (\ref{equ: BM-u-excess}), we have $$u_{excess}^*=\frac{1}{w}(\pi_{F}^*-d_{F}).$$ Thus, \eqref{equ: pi_L-basecase} follows.
Next,
\begin{align*}
\pi_{L}^*-d_{L}=\frac{1-w}{w}(\pi_{F}^*-d_{F}) = (1-w)u_{excess}
\end{align*}
Thus, \eqref{equ: pi_F-basecase} follows.
\qed

Since $0 < w < 1$, from \eqref{equ: pi_L-basecase}, \eqref{equ: pi_F-basecase},  $\pi_{L}^{*} \geq  d_L$ and $\pi_{F}^{*} \geq  d_F$ if and only if $u_{excess}^* \geq 0.$ 

Now, we can solve maximization \eqref{equ: BM-bargainging optimization-2} in two steps: 1) obtain the optimum $I_{L}^{*}$, $I_{F}^{*}$ by Theorem~\ref{thm: BM-bargaining optimization-3}, 2) obtain the optimum  $\tilde{s}^*, \theta^*$ by \eqref{equ: BM-bargaining-s} and \eqref{equ: BM-bargaining-theta}.


\begin{theorem}\label{thm: BM-bargaining optimization-3}
The optimum $(I_L^*, I_F^*)$ of  (\ref{equ: BM-bargainging optimization-2}) are also the optimum solutions of
\begin{equation}\label{equ: BM-bargainging optimization-3temp}
\begin{aligned}
\max_{I_{L}, I_{F}}\quad&u_{excess}\\
s.t.\quad&L_0\leq I_{L},\,\,0\leq I_{F}\leq I_{L}\\
&u_{excess} \geq 0
\end{aligned}
\end{equation}
\end{theorem}

\begin{remark}
Thus, the equilibrium-type  $(I_L^*, I_F^*)$ can be obtained by solving a maximization that
seeks to maximize the overall payoffs of the two SPs.
\end{remark}

\proof
From \eqref{equ: pi_L-basecase} and \eqref{equ: pi_F-basecase}
\[(\pi_{F}-d_{F})^{w}(\pi_{L}-d_{L})^{1-w}=w^{w}(1-w)^{1-w}u_{excess}.\]	
Since $0 < w < 1$, maximizing the objective function of Theorem~\ref{thm: BM-bargaining optimization-3},
is equivalent to maximizing $u_{excess}.$ 
  Right after defining $u_{excess}$, we have argued that $u_{excess}^*$ is  a function only of  $I_{F}^*, I_{L}^*, d_L, d_F$. Thus, $u_{excess}^*$ does not depend on $\tilde{s}^*, \theta^*$. We have already argued that $(\pi_{L}, \pi_{F})\geq (d_{L}, d_{F})$ is equivalent to $u_{excess}\geq0$. 
\qed

Since $u_{excess}$ is a function only of $I_L, I_F, d_L, d_F$ as noted right after its definition, the choice of $\tilde{s}, \theta$ does not affect $u_{excess}$. But,  $\tilde{s}^*, \theta^*$ must be determined so as to split $u_{excess}^* - d_L - d_F$ into $\pi_L^*, \pi_F^*$, as per  \eqref{equ: BM-F-payoff} and \eqref{equ: pi_F-basecase} \big(\eqref{equ: BM-L-payoff}, \eqref{equ: pi_L-basecase} follow from \eqref{equ: BM-F-payoff} and \eqref{equ: pi_F-basecase}\big). From \eqref{equ: BM-F-payoff} and \eqref{equ: pi_F-basecase},
  \[ \theta^* - \tilde{s}^* (I_F^*)^2   = wu_{excess}^* + d_F - n_F^* (p_F^* - c). \]
   When $I_F^* = 0$, $\theta^*$ is unique; otherwise,  there may be multiple values of $\tilde{s}^*, \theta^*$ which accomplish the above. When $I_F^* > 0$, we choose $\theta^* = 0$ and $\tilde{s}^*$ to satisfy the above equation. Our solution   utilizes  additional remuneration transfer  only when  SP$_F$ does not reserve any spectrum offered by SP$_L$ and thus that route for transfer of money between the SPs to ensure their  commensurate shares  is closed.  Thus, 
\begin{equation}\label{equ: BM-bargaining-s}
\begin{aligned}
\tilde{s}^{*}=\left\{\begin{aligned}
&\frac{1}{(I_{F}^{*})^{2}}(n_{F}^{*}(p_{F}^{*}-c)-d_{F}-wu_{excess}^{*})&\,& I_{F}^*>0\\
&\text{$s^*$ has no significance}&\,& I_{F}^*=0
\end{aligned}\right.
\end{aligned}
\end{equation}
\begin{equation}\label{equ: BM-bargaining-theta}
\begin{aligned}
\theta^{*}=\left\{\begin{aligned}
& 0 &\,& I_{F}^*>0\\
&d_{F}+wu_{excess}^{*} -  n_{F}^{*}(p_{F}^{*}-c) &\,& I_{F}^*=0
\end{aligned}\right.
\end{aligned}
\end{equation}

\begin{remark}\label{remark3}
Intuitively, as SP$_F$'s bargaining power ($w$) increases, he should get a larger share of the overall revenue. Thus, the marginal reservation fee he pays SP$_L$ ought to decrease and the additional remuneration he receives from SP$_L$ ought to increase. The analysis above confirms this intuition. From \eqref{equ: BM-bargainging optimization-3temp}, the equilibrium-type $I_L^*, I_F^*, u_{excess}^*$ do not depend on $w.$ Since  the equilibrium-type $n_L^*, n_F^*, p_L^*, p_F^*$ depend only on $I_L^*, I_F^*$, other than parameters such as $\Delta$,  $\tilde{s}^*$ (respectively, $\theta^*$)  is a linearly decreasing (respectively, increasing)  function of $w$, from \eqref{equ: BM-bargaining-s}~and~\eqref{equ: BM-bargaining-theta}.
\end{remark}



\subsection{Characterizing the equilibrium-type solutions}\label{sec: Un-outcome}
We now characterize the equilibrium type solutions. Unless otherwise mentioned, the proofs have been relegated to Appendix~\ref{Appendix: Unequal1}. 



\begin{theorem}\label{toprove}
 Let $|\Delta| <1$. The following holds for each equilibrium-type solution that may exist: $I_L^* = L_0$, and
  \begin{itemize}
  \item [(1)] If $-1<\Delta<0$, $I_{F}^{*} = L_0,$ and $  s^{*}$  is obtained by \eqref{equ: BM-bargaining-s}, and $\theta^*=0$.
\item[(2)] If $0<\Delta<1$, $I_{F}^{*} = 0$, $s^{*}$ has no significance, and $\theta^*$ is obtained by \eqref{equ: BM-bargaining-theta}.
\item[(3)] If $\Delta=0$, both the above constitute equilibrium-type solutions if there exists any equilibrium-type solution.
\end{itemize}
\end{theorem}

Assuming that the equilibrium-type solution exists, Theorem~\ref{toprove} gives the following insights.
 SP$_{L}$ always acquires minimum amount ($L_{0}$) of spectrum from a regulator. This is because the EUs must choose between the SP$_L$ and SP$_F$, and both determine their spectrum acquisition  together so as to maximize the overall profits and subsequently split their profits. The lack of competition leads to a collusive outcome in which they together opt for the minimum overall spectrum acquisition from the regulator. In contrast, when SP$_L$, SP$_F$ decide their spectrum acquisitions separately, $I_L^*$ exceeds the minimum mandated amount 
  (Theorem~1  of Part~\uppercase\expandafter{\romannumeral1}~\cite{Part1}). This happens because each SP seeks to maximize his profit through a sequence of non-cooperative games.


 The equilibrium-type solutions differ  in how the spectrum acquired from the regulator is split between SP$_L$ and SP$_F.$  This happens because the SPs decide the split of the acquired spectrum jointly to maximize their overall profits, which is accomplished if more EUs choose a SP that charges more. To ensure this, the more apriori popular SP retains the entire leased spectrum: 1) SP$_{F}$ if $v^{L}<v{^F}$, 2) SP$_L$   if $v^{L}>v^{F}$. If both  are equally popular apriori, i.e., $v^{L}=v^{F}$, both the above options constitute equilibrium-type solutions. Then,  even if the more apriori popular  SP charges a high price, more EUs would choose him because of his greater prior popularity and because he can offer  better quality of service  through the acquisition of the leased spectrum in its entirety. Thus, the more popular SP gets the lion share of subscription revenue, which he shares with the other. Thus, if SP$_{F}$ is more popular, he pays SP$_L$ $\tilde{s}^*I_{F}^{*2}$ amount ($I_F^* = I_L^*=L_0$ here); if SP$_{L}$ is more popular, he pays SP$_{F}$ $\theta^*$ amount ($I_F^* = 0,  I_L^*=L_0$ here). $\tilde{s}^* > 0$ in the first case, and   $\theta^*$ is $0$ and positive respectively in the two cases, as Theorem~\ref{thm: Un-outcome-1} will show.

  We now consider the degree of cooperation, i.e., $I_{F}^{*}/I_{L}^{*}$, which clearly equals  $0$ or $1$: these respectively arise   if SP$_{L}$ and SP$_F$ are respectively more apriori popular. Since the more apriori popular SP retains the entire leased spectrum (following Theorem~\ref{toprove} as explained in the previous paragraph), $I_{F}^{*}/I_{L}^{*}$ is $1$ if  $v^{L}<v{^F}$, i.e., if $\Delta = v_L - v_F < 0$,  and $I_{F}^{*}/I_{L}^{*}$ is $0$ if  $v^{L}> v{^F}$, i.e., if $\Delta > 0$.
  Thus, $I_{F}^{*}/I_{L}^{*}$ discontinuously transitions from $1$ to $0$ as $\Delta$ transitions from negative to positive, the transition occurring exactly at $\Delta = 0$ (Figure~\ref{fig-bargaining-degree of cooperation}).      If both SPs have equal priors, i.e., $\Delta = 0$, the degree of cooperation can be either $0$ or $1.$ In contrast, when SP$_L$, SP$_F$ decide their spectrum acquisitions separately,  $I_F^*$ can be between $0$ and $I_L^*$ (Theorem~1  of Part~\uppercase\expandafter{\romannumeral1}~\cite{Part1}).
  Figure~\ref{fig-bargaining-degree of cooperation} elucidates this distinction. The plot  for the individual spectrum acquisitions has been obtained from Theorem~1  of Part~\uppercase\expandafter{\romannumeral1}~\cite{Part1} considering at each $\Delta$,  $\tilde{s}$ to be  that which maximizes the sum of the disagreement payoffs. In this case, the jump in the degree of cooperation at a threshold value of $\Delta$ follows from Theorem~1~(2) of Part~\uppercase\expandafter{\romannumeral1}~\cite{Part1} directly. Figure 6 (left) of Part~\uppercase\expandafter{\romannumeral1}~\cite{Part1} also shows a plot for this case with a similar jump. 

   \begin{figure}
\begin{center}
  \includegraphics[width=2in]{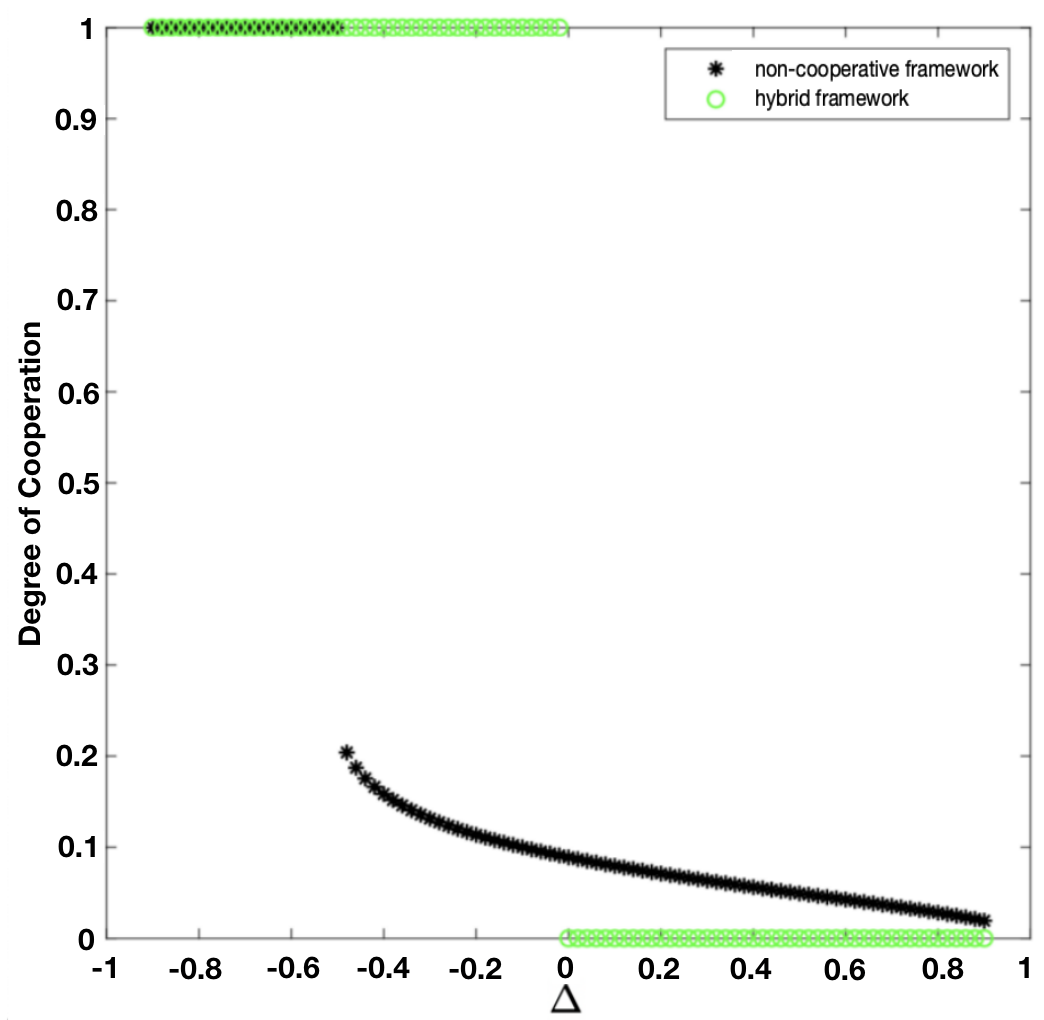}
  \caption{The degree of cooperation vs. $\Delta$.}
\label{fig-bargaining-degree of cooperation}
\end{center}
\end{figure}

Now, $\tilde{s}^{*}, \theta^*$ can be  obtained from \eqref{equ: BM-bargaining-s} and  \eqref{equ: BM-bargaining-theta} respectively,  $p_L^*, p_F^*, n_L^*, n_F^*$ can be obtained from Theorem~\ref{thm1Part1} leading to the following overall equilibrium-type solutions:

\begin{theorem}\label{thm: Un-outcome-1}
Let $|\Delta| < 1$:
\begin{itemize}
  \item [(1)] If $-1<\Delta<0$, if an  equilibrium-type solution exists, it is: $(I_{F}^{*}, I_{L}^{*})=(L_{0}, L_{0})$, $\tilde{s}^{*}$ is obtained by \eqref{equ: BM-bargaining-s}, $\theta^* = 0$, and
\begin{equation*}
\begin{aligned}
&p_{L}^{*}=c+\frac{1}{3}+\frac{\Delta}{3},\, p_{F}^{*}=c+\frac{2}{3}-\frac{\Delta}{3}\\
&n_{L}^{*}=\frac{1}{3}+\frac{\Delta}{3},\ n_{F}^{*}=\frac{2}{3}-\frac{\Delta}{3}
\end{aligned}
\end{equation*}

\item[(2)] If $0<\Delta<1$, if an   equilibrium-type solution exists, it is: $(I_{F}^{*}, I_{L}^{*})=(0, L_{0})$, $\tilde{s}^{*}$ is of no significance, $\theta^*$ is obtained by \eqref{equ: BM-bargaining-theta}, and
\begin{equation*}
\begin{aligned}
&p_{L}^{*}=c+\frac{2}{3}+\frac{\Delta}{3},\, p_{F}^{*}=c+\frac{1}{3}-\frac{\Delta}{3}\\
&n_{L}^{*}=\frac{2}{3}+\frac{\Delta}{3},\, n_{F}^{*}=\frac{1}{3}-\frac{\Delta}{3}
\end{aligned}
\end{equation*}

\item[(3)] If $\Delta=0$, if an equilibrium-type solution exists, the  equilibrium-type solutions are:
\begin{itemize}
  \item $(I_{F}^{*}, I_{L}^{*})=(0, L_{0})$, $\tilde{s}^{*}$ is of no significance, $\theta^*$ is obtained by \eqref{equ: BM-bargaining-theta},
\begin{equation*}
\begin{aligned}
&p_{L}^{*}=c+\frac{2}{3}=n_L^*+c,\, p_{F}^{*}=c+\frac{1}{3}=n_F^*+c
\end{aligned}
\end{equation*}
\item $(I_{F}^{*}, I_{L}^{*})=(L_{0}, L_{0})$, $\tilde{s}^{*}$ is obtained by \eqref{equ: BM-bargaining-s}, $\theta^*=0$,
\begin{equation*}
\begin{aligned}
&p_{L}^{*}=c+\frac{1}{3}=n_L^*+c,\, p_{F}^{*}=c+\frac{2}{3}=n_F^*+c.
\end{aligned}
\end{equation*}

\end{itemize}
\end{itemize}
\end{theorem}

Thus, considering only the values of $\tilde{s}^*, \theta^*$ given by \eqref{equ: BM-bargaining-s}, \eqref{equ: BM-bargaining-theta}, the equilibrium-type solution is easy to compute and unique when it exists, when $|\Delta| < 1$, with the only exception being at $\Delta = 0$, at which there are either $0$ or $2$ equilibria.  The  insights on $p_L^*, p_F^*, n_L^*, n_F^*$ are otherwise similar to those presented after Theorem 1 in Part~\uppercase\expandafter{\romannumeral1}~\cite{Part1}.

\begin{corollary}\label{cor: sum of payoffs independent of s*}
The sum of payoffs of each of the possible equilibrium-solutions presented in Theorem~\ref{thm: Un-outcome-1} is:
\begin{align}\label{equ: sum of payoffs independent of s*}
\pi^* = \pi_{L}^{*}+\pi_{F}^{*}=(1/3-|\Delta|/3)^2 + (2/3+|\Delta|/3)^2- \gamma L_0^2.	
\end{align}
\end{corollary}
\begin{proof}
First, from \eqref{equ: BM-L-payoff} and \eqref{equ: BM-F-payoff}, we have
\begin{align*}
\pi_L+\pi_F=n_L(p_L-c)+n_F(p_F-c)-\gamma I_L^2.
\end{align*}
From Theorem~\ref{thm: Un-outcome-1}, $n_L^*=p_L^*-c$, $n_F^*=p_F^*-c$, and $I_{L}^{*}=L_{0}$. Then inserting $n_L^*$, $n_F^*$, $p_L^*$, $p_F^*$, and $I_L^*$ into the above equation, we have the desired result.
\end{proof}
Again, assuming that  the equilibrium solution exists in each case, the total payoff of the SPs decreases with the minimum mandated amount of spectrum acquisition $L_0$.
This is expected as this reduction is in effect equivalent to relaxation of a constraint in a maximization, which increases the maximum value. Intuitively, the SPs increase their overall payoffs if they are allowed to get away with acquiring  really small amounts of spectrums; since the EUs must choose one of the SPs, the joint subscription revenues of the SPs is not affected as long as both SPs acquire small amounts of spectrum. The sum also  decreases with increase in the marginal reservation fee the central regulator charges.  The sum is maximized at $|\Delta| = 1$, i.e., when one of the two SPs is apriori substantially more  popular than the other, thus, he can attract most of the EUs despite charging a high amount. This enhances the overall subscription revenue. Note that the sum does not depend on the disagreement payoffs, and therefore does not depend on the marginal reservation fee the SP$_F$ pays the SP$_L$ in the event of a disagreement, i.e., the $s$ the market provides.

We provide a necessary and sufficient condition for the existence of equilibrium-type solutions, in terms of parameters $\Delta, \gamma, L_0$ and  disagreement payoffs $d_L, d_F$.
\begin{theorem}\label{thm:existence}
Let $|\Delta| < 1$. At least one equilibrium-type solution exists if and only if
\begin{align*}
\pi^*=&(1/3-|\Delta|/3)^2 + (2/3+|\Delta|/3)^2- \gamma L_0^2\\
\geq& d_L+d_F=d.
\end{align*}
\end{theorem}

\begin{remark}
\label{remark4}
The disagreement payoffs $d_L, d_F$ depend on the market-dependent marginal reservation fee $s$ the SP$_F$ pays the SP$_L$ in the event of disagreement. Thus, this $s$ only determines if an equilibrium-type solution  exists, but not its values. Clearly, such solutions do not exist for large $\gamma, L_0$, which is consistent with the insights developed in Remarks~\ref{remark1}, \ref{remark0}. In contrast, for $|\Delta| < 1$, the SPNE always exists, and is unique,   when the SPs decide everything individually (Theorem 1 of Part 1).
\end{remark}

 We now consider the EU-resource-cost metric introduced in the last paragraph of Section 2.1 of Part~\uppercase\expandafter{\romannumeral1}~\cite{Part1}, quantified as  $I_F/p_F + (I_L-I_F)/p_L$. We have from Theorem~\ref{thm: Un-outcome-1}:
 \begin{theorem}\label{thm: constant m*}
The EU-resource-cost metric in the SPNE is
\begin{align*}
\left\{\begin{aligned}
&L_0/(c+\frac{2}{3}-\frac{\Delta}{3})&\quad&-1<\Delta<0\\
&L_0/(c+\frac{2}{3})&\quad&\Delta=0\\
&L_0/(c+\frac{2}{3}+\frac{\Delta}{3})&\quad&0<\Delta<1
\end{aligned}\right..
\end{align*}
\end{theorem}
 Thus, the EU-resource-cost metric is clearly an increasing function of $L_0$. This is intuitive as in the SPNE the SPs together acquire exactly $L_0$ amount of spectrum. The Theorem goes beyond this intuition by identifying the exact nature of the dependence.   It increases with increase in $\Delta$, when $-1<\Delta<0$, reaches its maximum value at $\Delta = 0$, and decreases with increase in $\Delta$, when $0<\Delta<1.$ Thus, the EUs are best off, when the static factors are equal. 
It also decreases in increase in $c$,  since the SPs increase the access fee for the EUs with increase in $c$.

We now consider $|\Delta| \geq 1.$ As in Part~\uppercase\expandafter{\romannumeral1}~\cite{Part1}, this region is not of much interest due to the insurmountable difference between the prior preferences for the SPs. We show that in this case  equilibrium-type solutions exist only for very small values of $L_0$. Since these solutions provide $I_L^* = L_0$, even the solutions are of limited practical utility. We state the results for completeness.   Let $s, \delta$ constitute the parameters that provide the disagreement payoff  (from the sequential game of Part~\uppercase\expandafter{\romannumeral1}~\cite{Part1}. Let $\gamma < s.$

\begin{theorem}\label{thm: Un-outcome-4}
If $\Delta\leq-1$ or $\Delta\geq1$, the equilibrium-type solutions exist.
\begin{itemize}
  \item [(1)] If $\Delta\leq-1$ and $L_{0}\leq \frac{1}{\sqrt{2s}}$, the equilibrium-type solutions are: $I_{L}^{*}=L_{0}$,  $I_{F}^{*}\in[0,L_{0}]$, $s^{*}$ is obtained by \eqref{equ: BM-bargaining-s}, $\theta^*$ is obtained by \eqref{equ: BM-bargaining-theta}, and
\begin{equation*}
\begin{aligned}
&p_{L}^{*}=p_{F}^{*}+\Delta+1,\, c+1\leq p_{F}^{*}\leq c-\Delta-1,\\
&n_{L}^{*}=0,\, n_{F}^{*}=1,
\end{aligned}
\end{equation*}
If $L_{0} > \frac{1}{\sqrt{2s}}$, no equilibrium-type solution exists.
\item[(2)] If $\Delta\geq1$ and $L_{0}\leq \delta$, the equilibrium-type solutions are: $I_{L}^{*}=L_{0}$, $I_{F}^{*}\in[0,L_{0}]$, $s^{*}$ is obtained by \eqref{equ: BM-bargaining-s}, $\theta^*$ is obtained by \eqref{equ: BM-bargaining-theta}, and
\begin{equation*}
\begin{aligned}
&p_{F}^{*}=p_{L}^{*}-\Delta,\, c+1\leq p_{L}^{*}\leq c+\Delta.\\
&n_{L}^{*}=1,\, n_{F}^{*}=0.
\end{aligned}
\end{equation*}
If $L_{0}>\delta$, no equilibrium-type solution exists.
\end{itemize}
\end{theorem}

\begin{remark}
When $\Delta\leq-1$, Theorem $2$ of Part~\uppercase\expandafter{\romannumeral1}~\cite{Part1} shows that the disagreement payoffs are attained when SP$_L$ acquires $\frac{1}{\sqrt{2s}}$
resource. If $L_0 > \frac{1}{\sqrt{2s}}$,  equilibrium-type solution does not exist per the intuitions in Remarks~\ref{remark1}.  More specifically, in this case,  the SPs together attain payoffs lower than the total disagreement payoffs, as they are forced to acquire greater amounts of spectrum than what they did for acquiring their disagreement payoffs. This does not increase the total subscription revenue as the EUs must choose one of the SPs, but increases the total cost incurred in spectrum acquisition from the  regulator. Thus, the aggregate excess payoff is negative. Hence  there is no equilibrium-type solution. If $\Delta\geq1$ and $L_{0} > \delta$, equilibrium-type solutions do not exist for similar reasons which follow from an application of Theorem $2$ of Part~\uppercase\expandafter{\romannumeral1}~\cite{Part1}.
\end{remark}

\subsection{Numerical results}
\label{numerical}

We numerically investigate the payoffs, the degree of cooperation, the investment levels, and the split of EUs to the SPs for $|\Delta| < 1$ and  different values of other  parameters. We set $\gamma=0.5$, $c=1$, $w=0.2$, and consider two cases: 1)  $\Delta=-0.5$; 2) $\Delta=0.5$. SP$_F$ is apriori more popular in the first, and SP$_L$ in the second. We refer to the sum of equilibrium-type solution payoffs of the SPs as $\pi^*$, and disagreement payoffs of the SPs as $d$.

We first examine the condition for existence of equilibrium-type solutions, given in  Theorem~\ref{thm:existence}, by varying $L_0$ between $[0.1,1]$, and different values of $s$ used to obtain the disagreement payoffs.  As expected from Theorem~\ref{thm:existence}, Figure~\ref{fig-bargaining1} show that $\pi^*$ decreases with $L_0$, and does not depend on $s.$ As mentioned in  Remark~\ref{remark4}, $d$ depends on $s$, and from Theorem~\ref{thm:existence} $d$ does not depend on $L_0$. Thus, the plots of $d$ are parallel to the x-axis in Figure~\ref{fig-bargaining1}. We note that $d$ initially increases with $s$, and then reaches its maximum value, at $s = s_{best} = 23.9$ and subsequently decreases. We consider $s =  0.8, 1, 1.2, s_{best}.$ Figure~\ref{fig-bargaining1} show the region in which  $d \leq \pi^*$, for different values of $s$, it is the region of existence of equilibrium-type solutions as per  Theorem~\ref{thm:existence}. For a given $s$, we do not plot $\pi^*$, once it falls below $d$; thus the curves for $\pi^*$ corresponding to a specific $s$ stop whenever they meet the $d$ for that $s.$     The region in which equlibrium-type solutions exist is smallest at $s=s_{best}$ and much larger at $s= 0.8.$ Referring to Corollary~\ref{cor: sum of payoffs independent of s*} and Theorem~\ref{thm:existence},  in this region   $\pi^*-d$  shows the gain in overall payoffs of the SPs through joint decision on spectrum acquisitions. The gain is naturally the smallest  at $s = s_{best}$, but significant at other values of $s.$


\begin{figure}[h!]
\begin{center}
  \includegraphics[width=1.7in, height=1.7in]{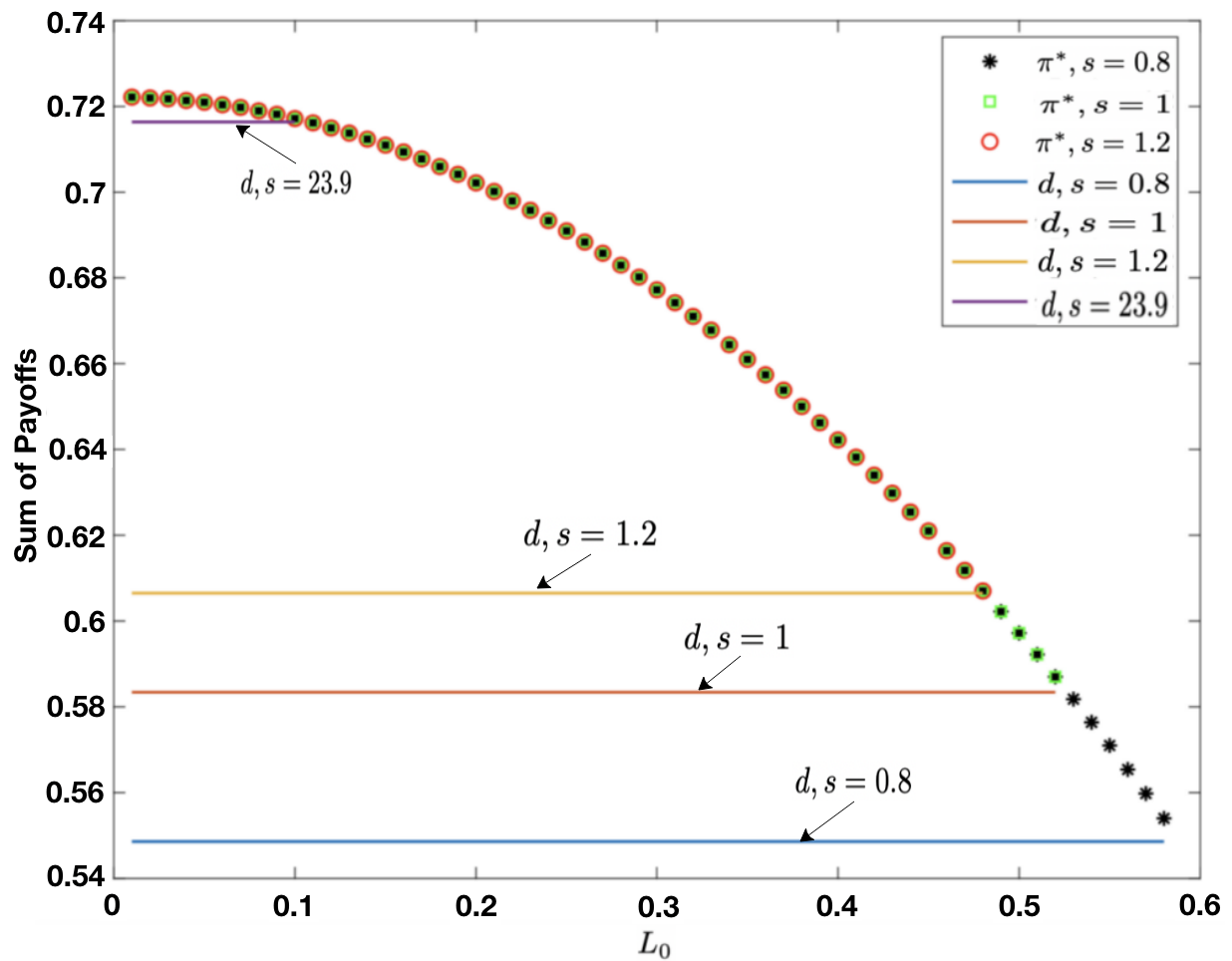}
  \includegraphics[width=1.7in, height=1.7in]{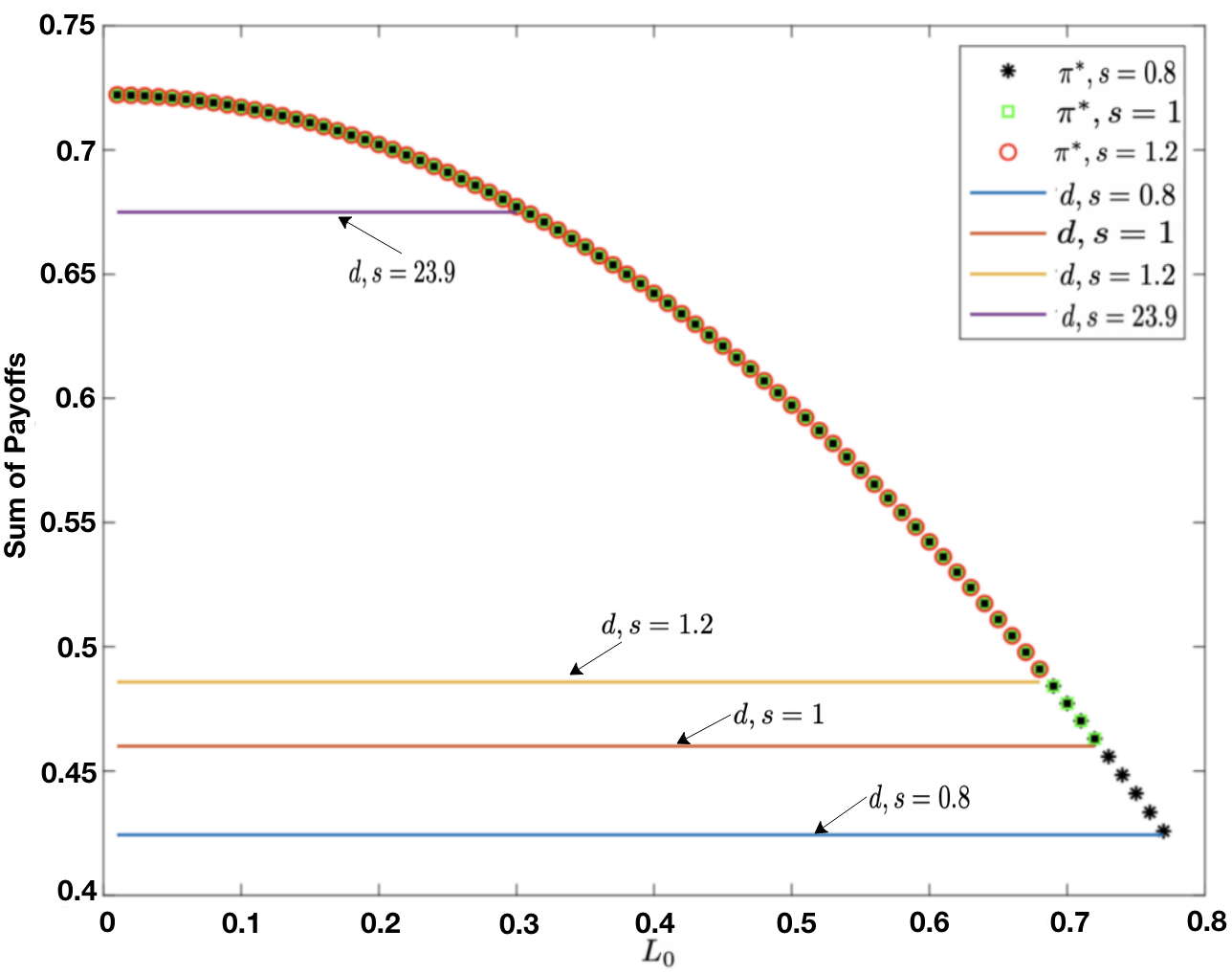}
  \caption{Sum of equilibrium-type solution payoffs of SPs, $\pi^*$, when $\Delta=-0.5$ (left) and when $\Delta=0.5$ (right) vs. $L_{0}$, sum of disagreement payoffs of SPs is $d$.}
   \label{fig-bargaining1}
\end{center}
\end{figure}

\begin{figure}
\begin{center}
  \includegraphics[width=1.7in]{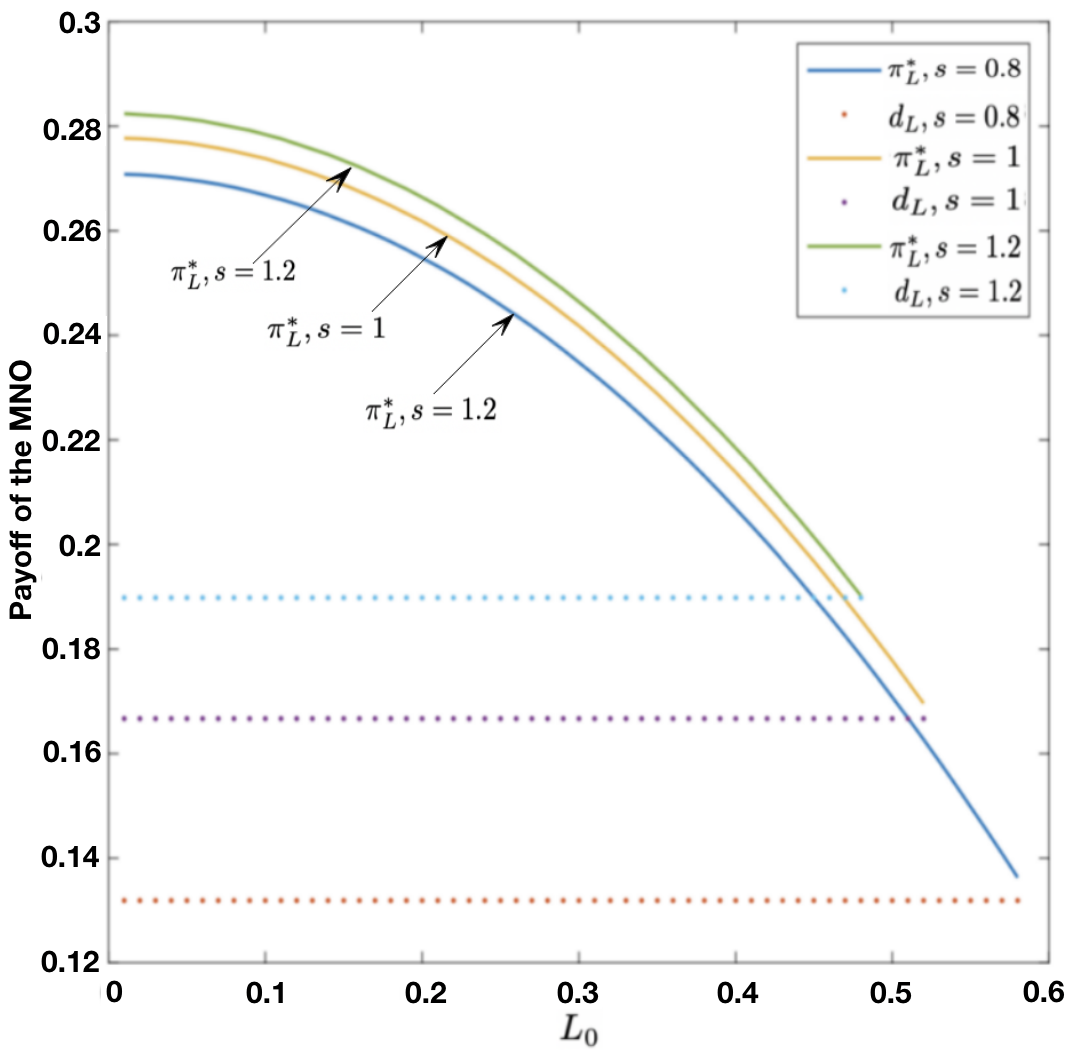}
  \includegraphics[width=1.7in]{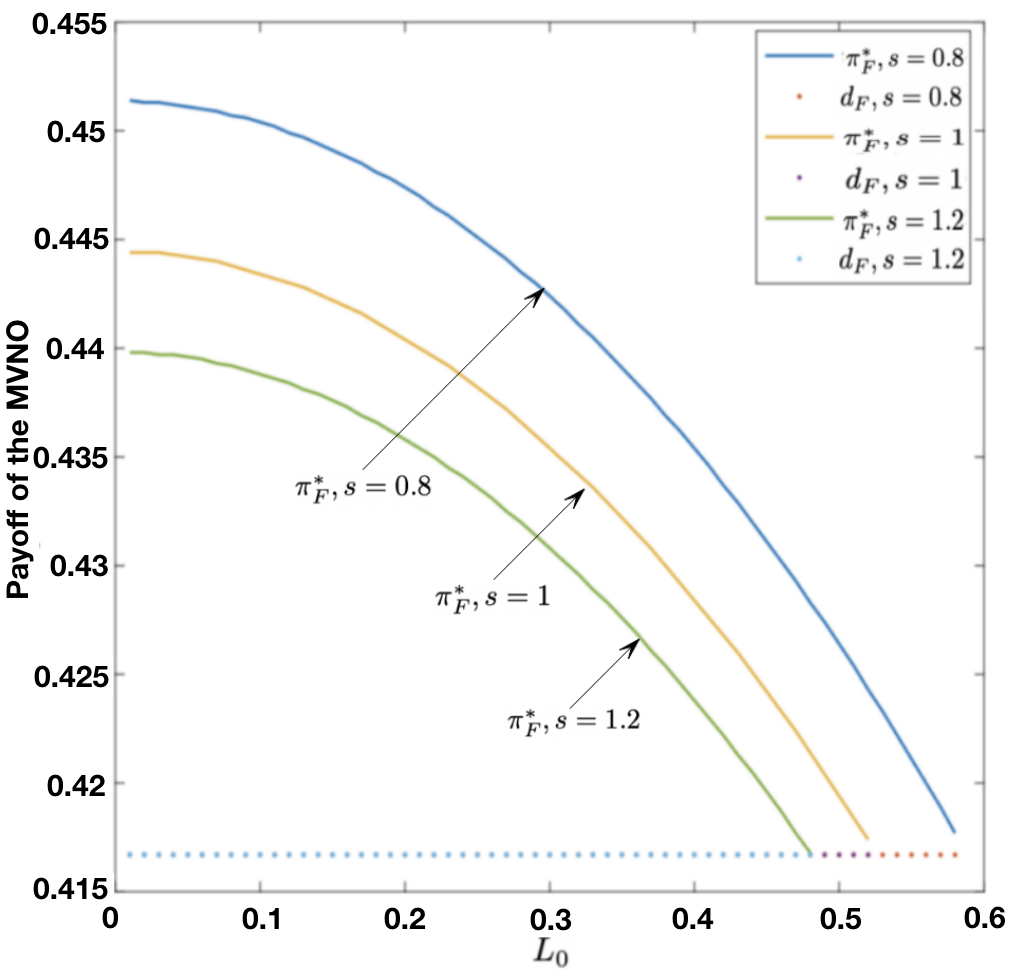}
  \caption{Equilibrium-type solution and disagreement payoffs of SP$_{L}$ (left), SP$_{F}$ (right) when $\Delta=-0.5$ vs. $L_{0}$.}
\label{sepa_payoff1}
\end{center}
\end{figure}

\begin{figure}
\begin{center}
  \includegraphics[width=1.7in]{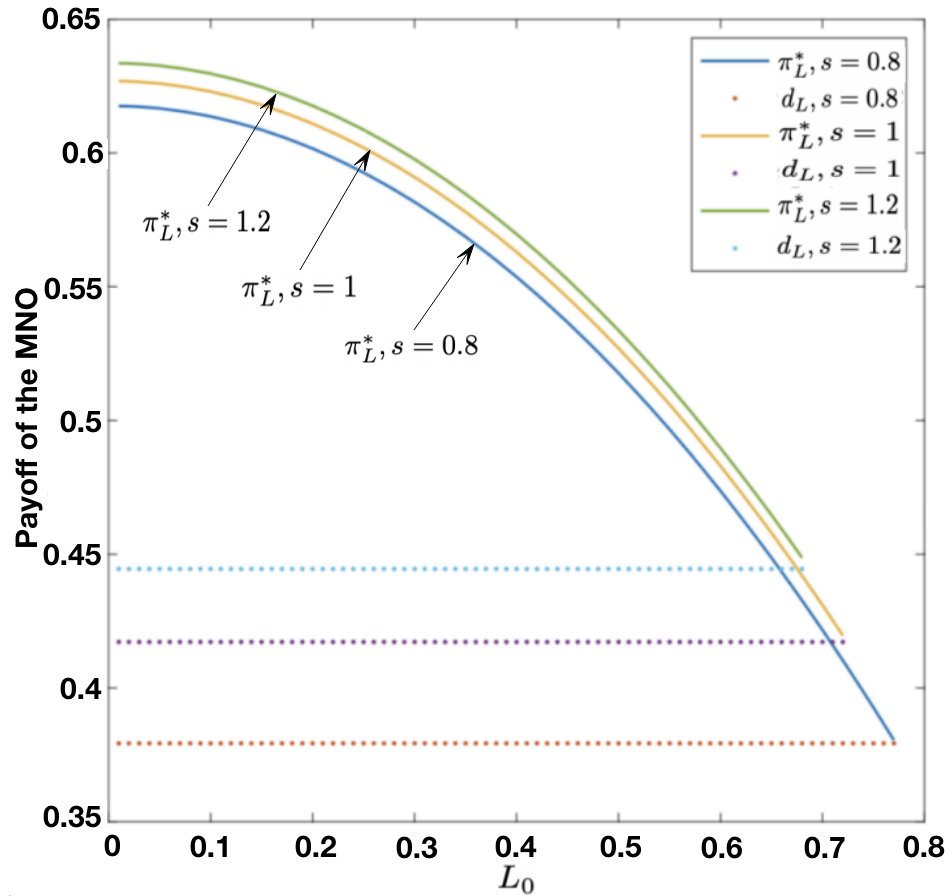}
  \includegraphics[width=1.7in]{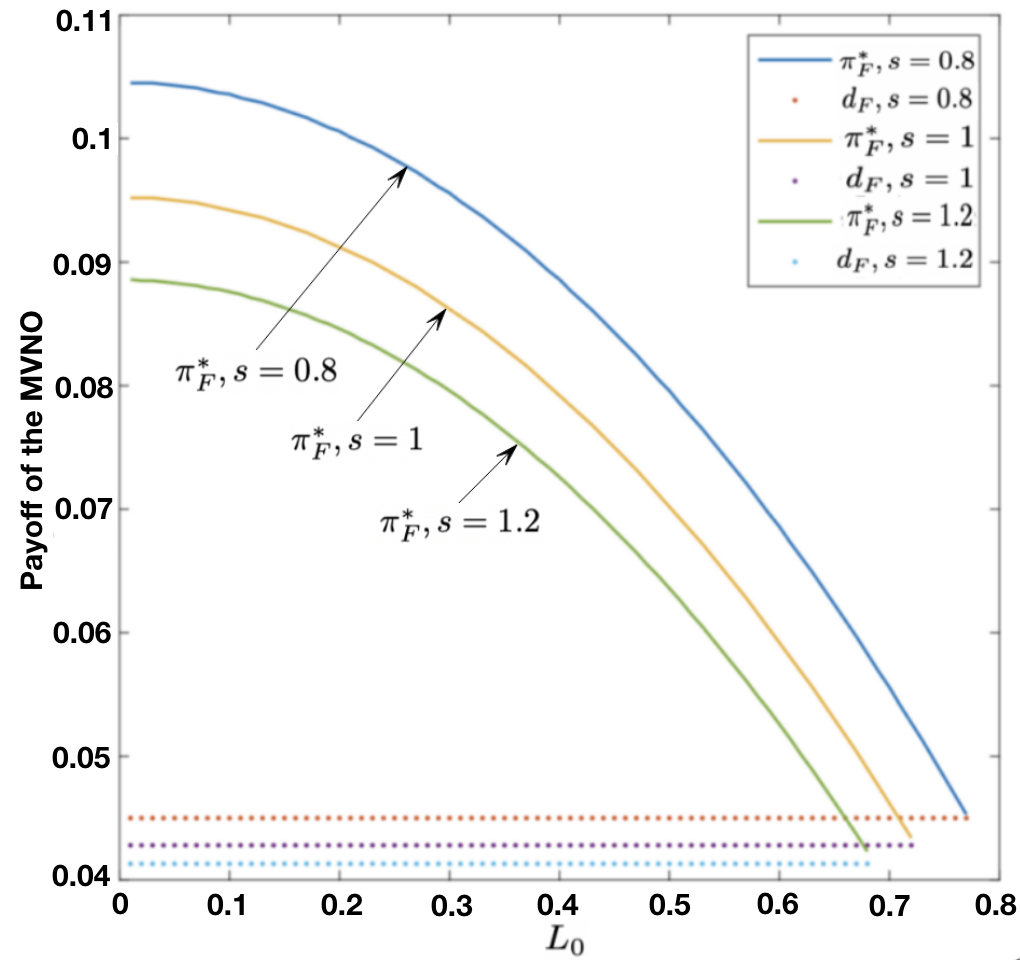}
  \caption{Equilibrium-type solution and disagreement payoffs of SP$_{L}$ (left),  SP$_{F}$ (right) when $\Delta=0.5$ vs. $L_{0}$.}
\label{sepa_payoff2}
\end{center}
\end{figure}

Figures~\ref{sepa_payoff1}, \ref{sepa_payoff2} demonstrate the payoff gain of each SP due to joint decision on spectrum acquisition (Remark~\ref{remark4repeat}). For   $\Delta \in \{-0.5, 0.5\}$, the figures  plot their disagreement payoffs,  and equilibrium-solution payoffs, in the region that equilibrium solutions exist, as given by Theorem~\ref{thm:existence} and illustrated in Figure~\ref{fig-bargaining1}.
 The payoffs under the equilibrium solutions now depend on $s$ as the disagreement payoffs depend on $s$ \big(Equations~\eqref{equ: pi_L-basecase} and \eqref{equ: pi_F-basecase}\big).    First consider  $\Delta=-0.5$.  
 Figure~\ref{sepa_payoff1} shows that
 the payoff of SP$_{L}$ (SP$_{F}$, respectively) increases (decreases, respectively) with $s$. Both payoffs decrease with $L_{0}$. Also,  the gain in payoff for each SP beyond his disagreement payoff, due to joint decision on spectrum,  is considerable for low $L_0$, but decreases as $L_0$ increases (the SPs are forced to acquire higher amounts of spectrum for large $L_0$ to deliver the minimum quality of service mandated by the  regulator).  The payoff of SP$_{F}$ is higher than that of SP$_{L}$ in this case because SP$_F$ is more apriori popular (as $\Delta < 0$).
 When $\Delta=0.5,$ Figure~\ref{sepa_payoff2} shows 
  that $\pi_{L}^{*}>\pi_{F}^{*}$, i.e., SP$_{L}$ has higher payoff in this case, which is intuitive as SP$_L$ is more apriori popular ($\Delta > 0$).  The observations are otherwise similar to those for Figure~\ref{sepa_payoff1}.  

\begin{figure}
\begin{center}
  \includegraphics[width=1.7in]{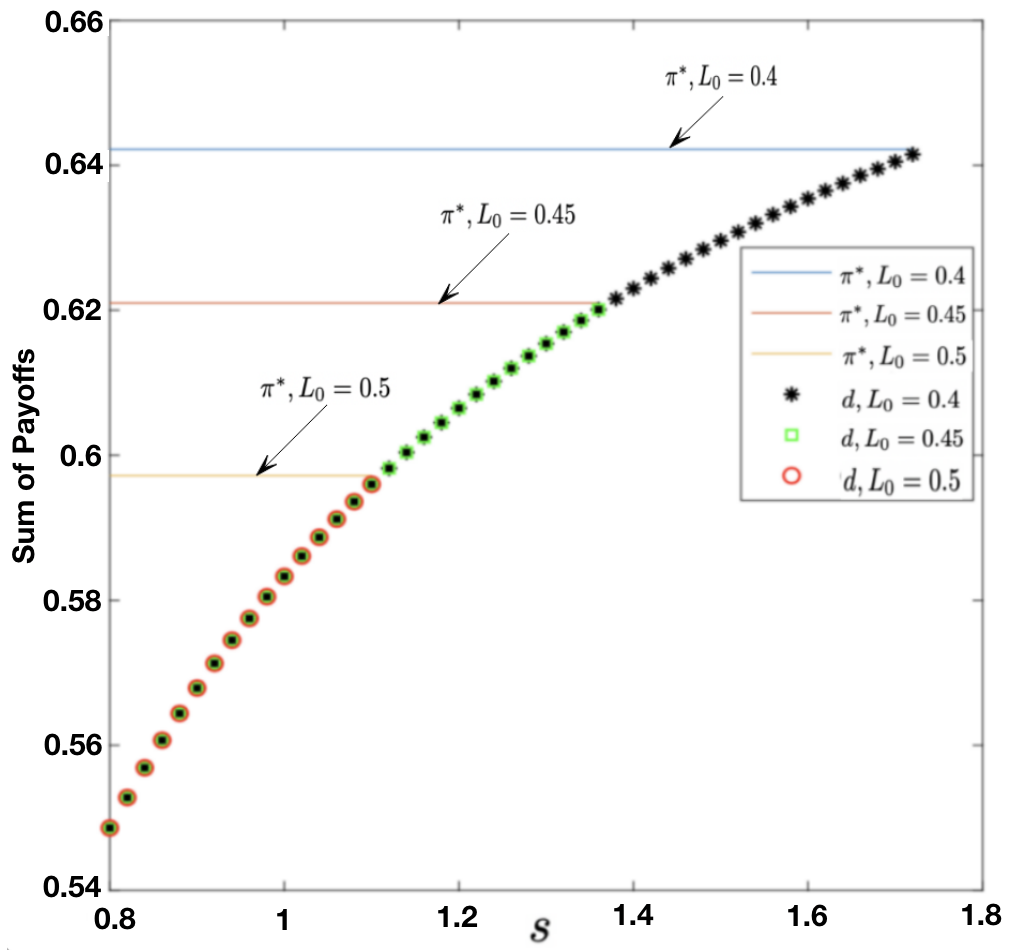}
   \includegraphics[width=1.7in]{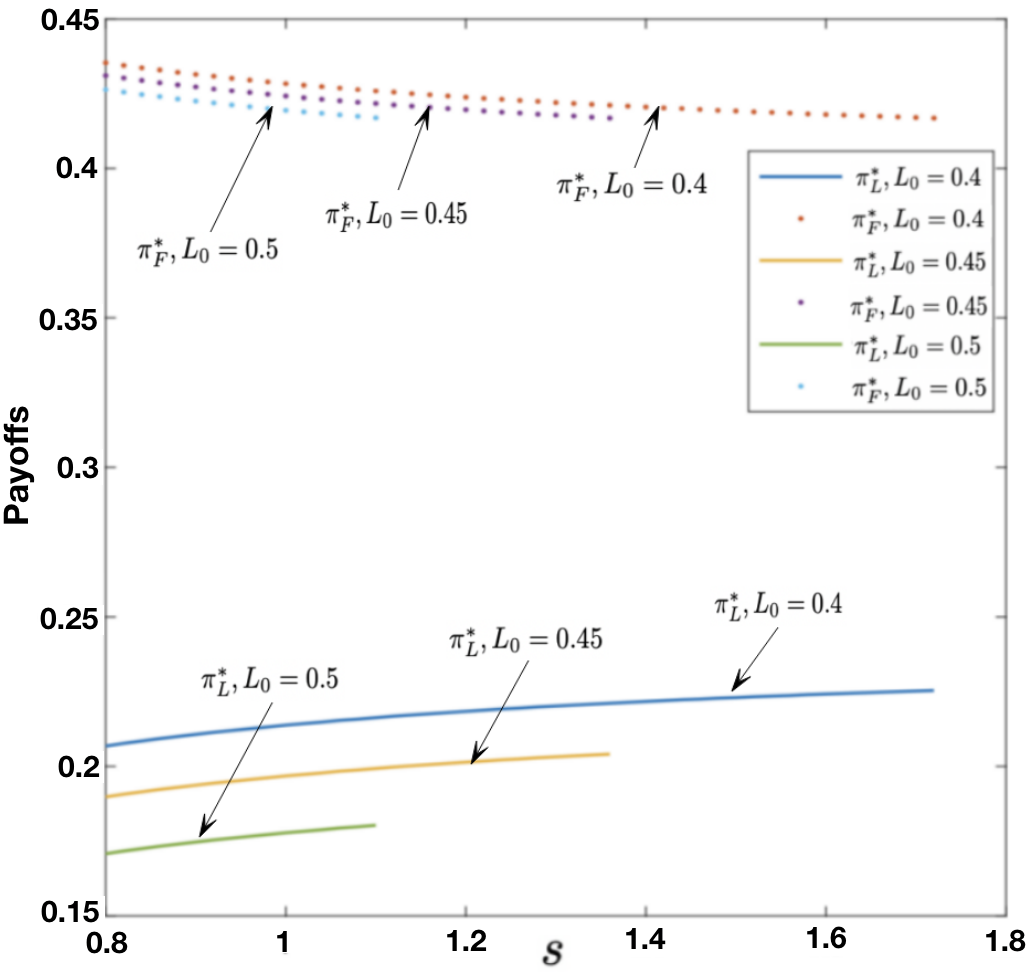}
  \caption{Sum of equilibrium-type solution payoffs $\pi^*$, sum of disagreement payoffs $d$ (left), individual payoffs (right) vs. $s$.}
\label{fig-bargaining-payoff-L}
\end{center}
\end{figure}

The above observations for existence of equilibrium-type solutions and the collective and individual gains in payoffs of the SPs due to joint decision on spectrum acquisition may be reinforced by plotting $\pi^*, d, \pi^*_L, \pi^*_F$ as functions of $s$ for few fixed $L_0$s (Figure~\ref{fig-bargaining-payoff-L}). We consider only $\Delta = -0.5$ here, and $L_{0}=0.4,0.45,0.5$.  Now, in the left figure,
 plots of $\pi^*$ are parallel to the $x$-axis, while $d$ increases with increase in $s$ in the range considered, $s \in [0.8,2]$. We plot $\pi^*$ only in the region in which the equilibrium-type solutions exist, i.e., where $\pi^* \geq d$. This plot also quantifies the gains in collective payoffs by showing how much the flat curves exceed the increasing one, in the region in which they are plotted. The figure in the right show that the payoff of each SP decreases with $L_{0}$, and payoff of SP$_{F}$ (SP$_{L}$, respectively) decreases (increases, respectively) with $s$. The payoff of SP$_{F}$ is higher than that of SP$_{L}$, which is intuitive as SP$_F$ is apriori more popular in this case.

Our numerical computations thus far reveal that the cooperation in form of joint decisions on spectrum acquisitions benefits the SPs by enhancing their collective and individual payoffs. We now investigate how this enhanced cooperation between the SPs affects the EUs.

\begin{figure}
\begin{center}
  \includegraphics[width=1.7in, height=1.7in]{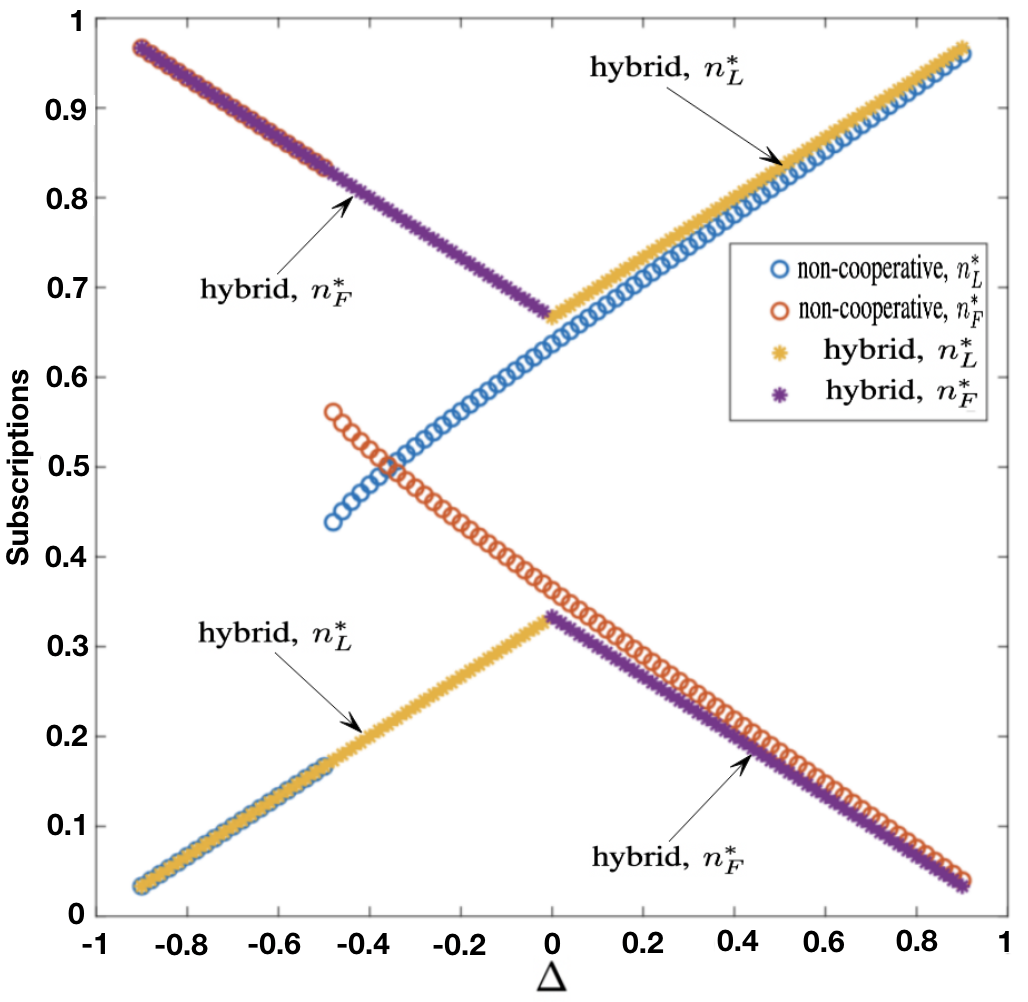}
  \includegraphics[width=1.7in,height=1.7in]{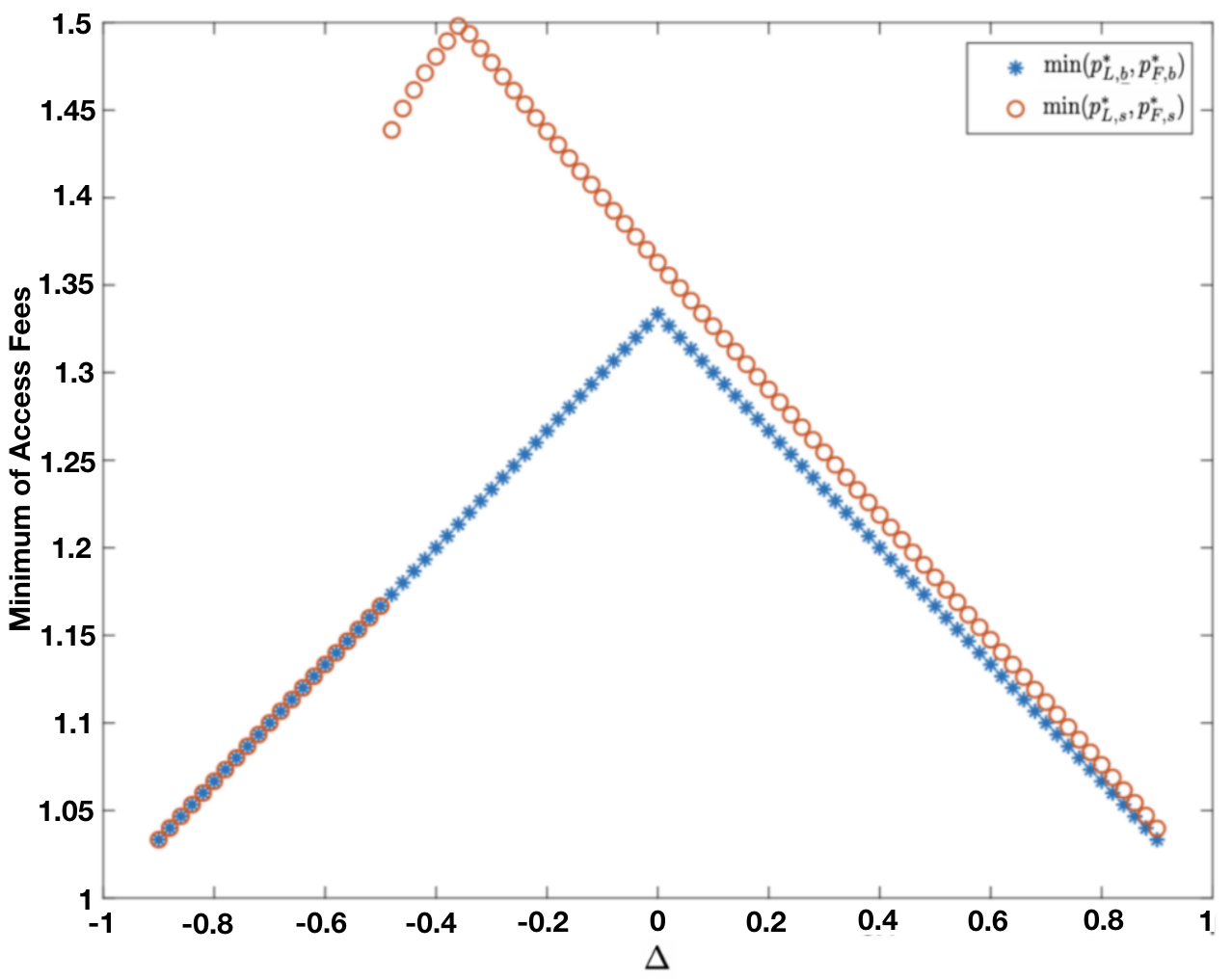}
  \caption{The subscriptions (left), the access fees (right) vs. $\Delta$.}
\label{fig-bargaining-subscrptionsprices}
\end{center}
\end{figure}
Now we investigate the subscriptions and access fees when the reservation fee is $s_{best}$, with $\Delta$ varying in $(-1,1)$. From Theorem~\ref{thm: Un-outcome-1},  subscriptions $n_L^*$, $n_F^*$ only depend on $\Delta$ (and are independent of $s$ and $L_{0}$).
Figure~\ref{fig-bargaining-subscrptionsprices} plots the subscriptions (left). From Theorem~\ref{thm: Un-outcome-1}, $p_{L}^{*}=n_{L}^{*}+c$ and $p_{F}^{*}=n_F^*+c$, so the equilibrium-type subscriptions and access fees exhibit similar behaviors. For $\Delta < 0$, i.e.,  when     SP$_{F}$ is apriori more popular, $n_F^* > n_L^*$ for the equilibrium-type solution under joint decision on spectrum acquisition. The difference $n_F^*-n_L^*$  increases as $\Delta$ reduces, when $\Delta < 0.$ The reverse is observed when $\Delta > 0.$ Since $p_F^* - p_L^* = n_F^* - n_L^*$ throughout, more EUs choose the SP that charges higher; this choice is clearly induced by how the SPs share between them the spectrum $I_L$ that SP$_L$ acquires. In some way, this benefits the SPs, enhancing their overall revenue, and harms the EUs by motivating them to pay more. The conclusion is however nuanced as the EUs choose the more expensive option, voluntarily, and only because that option provides better quality of service by retaining the acquired spectrum in its entirety, and was also apriori more popular. The choice is therefore guided, rather than enforced, by having the more apriori popular SP retain the acquired spectrum. When the SPs separately decide their spectrum acquisitions, the trends are similar, through the differences between the subscriptions, and therefore the access fees, is less pronounced. The spectrum is more evenly shared between the SPs (Figure~\ref{fig-bargaining-degree of cooperation}), leading to lower access fees and lower qualities of service for more EUs.

In  Figure~\ref{fig-bargaining-subscrptionsprices}~(right), we plot the minimum of access fees of SPs in both frameworks: $\min(p_{L,b}^*,p_{F,b}^*)$ ($\min(p_{L,s}^*,p_{F,s}^*)$, respectively) represent the minimum access fees when the SPs decide spectrum acquisitions  jointly (separately, respectively). The minimum access fee represents the least cost an EU might incur. The minimum is clearly equal or lower for the joint decision case. Thus, joint decisions of the SPs benefits the EUs by providing them cheaper access. But, as we have noted in the previous paragraph, more EUs are induced to select the more expensive option by having it provide the better quality of service and choosing the more popular of the two SPs to do so. Thus, in one perspective,   the EUs gain  due to enhanced coordination between the SPs, while they lose in another perspective.

We now plot the EU-resource-cost metric quantified in Theorem~\ref{thm: constant m*}.
Figure~\ref{fig-bargaining-ad1} shows that for both $\Delta>0$ or $\Delta<0$, for some values of $L_0$,  this metric is higher under the joint decision  and lower for the rest. As Theorem~\ref{thm: constant m*} shows, the EU-resource-cost metric is a linear function of $L_0$ under the joint decision, and therefore is higher or lower than that for individual decisions (as in Part I~\cite{Part1})  depending on the value of $L_0.$

\begin{figure}
\begin{center}
  \includegraphics[width=1.7in, height=1.7in]{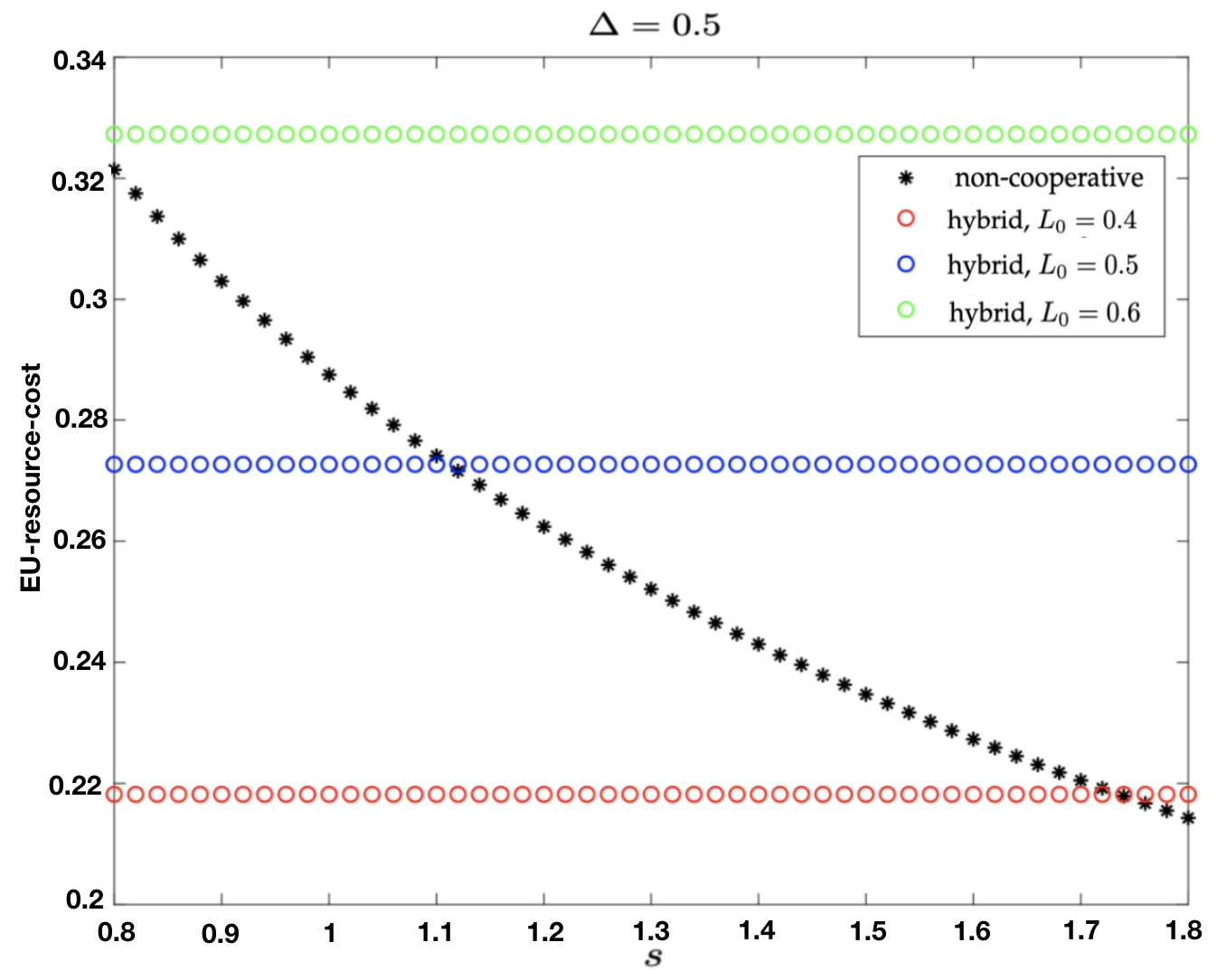}
  \includegraphics[width=1.7in,height=1.7in]{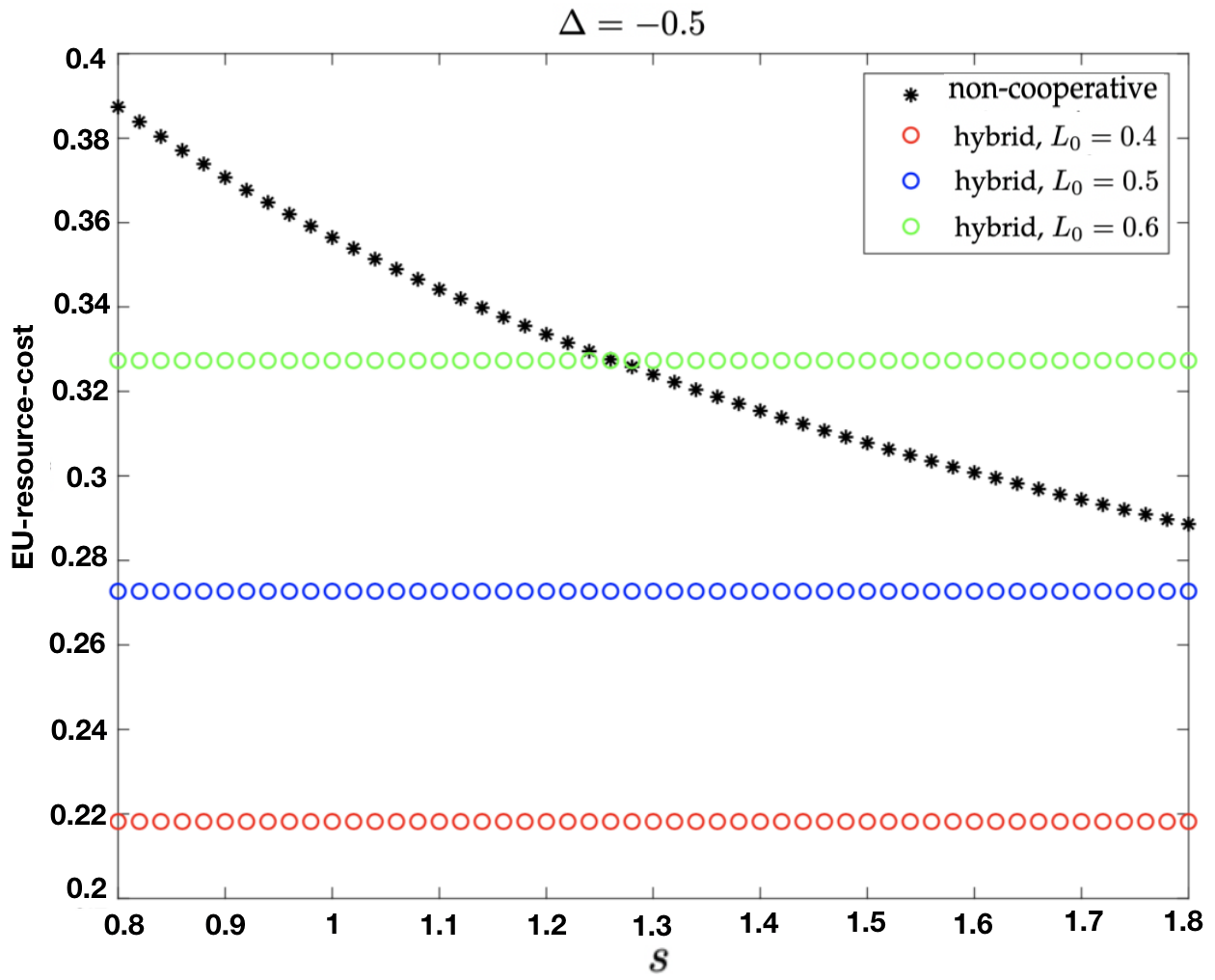}
  \caption{The EU-resource-cost vs. $s$ under $\Delta=0.5$ (left) and $\Delta=-0.5$ (right).}
\label{fig-bargaining-ad1}
\end{center}
\end{figure}

\section{EUs with outside option}\label{sec: outside option}

We now generalize our framework to allow the EUs from the common pool the MVNO and the MNO are competing over,  to choose a SP outside the system,  if  neither the MVNO nor the MNO  offers a desirable combination of access fee and quality of service. The SPs outside the system are collectively referred to as ``outside option''; we do not consider the strategies of these SPs.  Thus, the MVNO and the MNO may both experience an  attrition in their subscriptions.  We also allow the MVNO and the MNO to have exclusive additional customer bases to draw from depending on their individual spectrum acquisition and offered prices.  We introduce these modifications through demand functions we describe next.

Similar to equations (10), (11) in Part~\uppercase\expandafter{\romannumeral1}~\cite{Part1}, we define
the fraction\footnote{Recall that as in Part~\uppercase\expandafter{\romannumeral1}~\cite{Part1} $n_L, n_F$ are the fraction of EUs from the common pool who subscribe to the EUs, while $\tilde{n}_L, \tilde{n}_F$ may be the fractions or actual numbers of subscriptions, considering the attrition to the outside option and the additions from the exclusive customer bases. Only scale factors would change in the expressions for $\tilde{n}_L, \tilde{n}_F$ and the payoffs depending on if $n_L, n_F$ are fractions or actual numbers.} of EUs with each SP as
\begin{equation}\label{equ: Out-demand}
\begin{aligned}
&\tilde{n}_{L}=\alpha\big(n_{L}+\varphi_{L}(p_{L}, I_{L})\big),\\
&\tilde{n}_{F}=\alpha\big(n_{F}+\varphi_{F}(p_{F}, I_{F})\big),
\end{aligned}
\end{equation}
where
\begin{equation}
\begin{aligned}\label{demandfunction}
&\varphi_{L}(p_{L}, I_{L})=k-p_{L}+b(I_{L}-I_{F}),\\
&\varphi_{F}(p_{F}, I_{F})=k-p_{F}+bI_{F}
\end{aligned}
\end{equation}
and $\alpha>0$, $k$ and $b$ are constants.

We also define
$g(I_{L})=\frac{b}{15}I_{L}+\frac{1}{15}-\frac{c}{3}+\frac{k}{3}$, $f(I_{L})=\frac{1}{5I_{L}}+\frac{b}{5}>0$.

We characterize the equilibrium-type solutions in Section~\ref{analysis}, and examine its salient properties through numerical computations in Section~\ref{numericaloutside}.

\subsection{The equilibrium-type solution}
\label{analysis}
Our goal here is to examine if the availability of the outside option deters the collusive outcome by which the SPs  acquire the minimum mandated amount of spectrum from the central regulator. We focus on the region in which at least one interior equilibrium-type solution, i.e., $0 < n_L, n_F < 1$ exists, and show that this is indeed the case.  
 The  proofs are given in Appendix \ref{APP: outside option}.

\begin{theorem}\label{thm: out-outcome-sectionA}
Let $\Delta = 0$. Either there is no interior equilibrium-type solution, or there are two  interior equilibrium-type solutions. They are:
\begin{itemize}
\item[(1)] $I_{L,1}^{*}$ is a solution of
\begin{equation*}
\footnotesize
\begin{aligned}
\footnotesize
\max_{I_{L}}\,\,& 2\alpha g^{2}(I_{L})+2\alpha(f(I_{L})I_{L}+g(I_{L}))^{2}-\gamma I_{L}^{2}\\
s.t\,\,&L_0\leq I_{L}\,\,
\end{aligned}
\end{equation*}
$I_{F,1}^{*}=I_{L,1}^{*}$, $\tilde{s}^{*}$ is obtained by \eqref{equ: BM-bargaining-s}, and $\theta^*=0$.
\item[(2)]
$p_{L,1}^{*}=\frac{1}{15}+\frac{2c}{3}+\frac{k}{3}+\frac{bI_{L}^{*}}{15},\,
p_{F,1}^{*}=\frac{4}{15}+\frac{2c}{3}+\frac{k}{3}+\frac{4bI_{L}^{*}}{15}$.
\item[(3)] $\tilde{n}_{L,1}^{*}=\frac{2}{15}+\frac{2k}{3}+\frac{2bI_{L}^{*}}{15}-\frac{2c}{3},\, \tilde{n}_{F,1}^{*}=\frac{8}{15}+\frac{2k}{3}-\frac{2c}{3}+\frac{8bI_{L}^{*}}{15}$.
\end{itemize}
and
\begin{itemize}
\item[(1)] $I_{L,2}^{*} = I_{L,1}^*, I_{F,2}^{*}=0$, $\tilde{s}^{*}$ has no significance, and $\theta^*$ is obtained by \eqref{equ: BM-bargaining-theta}.
\item[(2)] $p_{L,2}^{*}=p_{F,1}^{*}$, $p_{F,2}^{*}=p_{L,1}^{*}$.
\item[(3)] $\tilde{n}_{L,2}^{*}=\tilde{n}_{F,1}^{*}$, $\tilde{n}_{F,2}^{*}=\tilde{n}_{L,1}^{*}$.
\end{itemize}
\end{theorem}

We provide a necessary and sufficient condition for the existence of equilibrium-type solutions, in terms of parameters $\alpha, \gamma, I_{L,1}^*$ and  disagreement payoffs $d_L, d_F$.
\begin{theorem}\label{thm:existence_outside}
Let $\Delta = 0.$ Interior equilibrium-type solutions exist if and only if
\begin{align*}
\pi^*=&2\alpha g^{2}(I_{L,1}^*)+2\alpha(f(I_{L,1}^*)I_{L,1}^*+g(I_{L,1}^*))^2- \gamma I_{L,1}^{*2}\\
\geq&d_L+d_F=d,\quad \text{and}\quad I_{L}^{*}<\frac{4}{b}.
\end{align*}
\end{theorem}

\begin{remark}\label{remark9}
Solutions do not exist for large $\gamma$ or small  $\alpha$, following the insights developed in Remarks~\ref{remark1}, \ref{remark0}.
\end{remark}

The equilibrium-type solutions are easy to compute as they involve optimization in one decision variable and closed-form expressions. They are not unique, unlike in Part~\uppercase\expandafter{\romannumeral1}~\cite{Part1} (Theorem 7). 

Our numerical computations would reveal that $I_L^*$  exceeds $L_0$ in some cases. Thus,  the deterrent of overall attrition and the incentive of increasing subscription from the exclusive additional bases, induce the SPs to acquire more spectrum than the minimum mandated amount, even when they are jointly deciding the acquisition amounts. Note that $p_{L}^{*}, p_{F}^{*}, n_{L}^{*}, n_{F}^{*}$ are linear increasing function of $I_{L}^{*}$. Thus, the SPs can increase both their subscriptions and access fees by acquiring greater overall spectrum $I_L^*$ from the regulator.  Like in the base case, $I_F^* \in \{0, I_L^*\}$, and thus the degree of cooperation is either $0$ or $1.$ This is in contrast to the equivalent case in Section~3 Part~\uppercase\expandafter{\romannumeral1}~\cite{Part1} (eg, Figure 5) which show that the degree of cooperation can assume values between $0$ and $1.$  Then, we consider the competition between SPs, i.e., the subscription $n_{L}^{*}$ and $n_{F}^{*}$. The subscriptions $n_{L}^{*}$ and $n_{F}^{*}$ are constant if there exists no outside option (Theorem~\ref{thm: Un-outcome-1}~(3)); but $\tilde{n}_{L}^{*}$ and $\tilde{n}_{F}^{*}$ change with the spectrum acquisition level of SP$_{L}$, $I_{L}^{*}$, if there exists an outside option.

We can write the first equilibrium-type solution as
\begin{align*}
\tilde{n}_{L}^*=&\frac{1}{5}+\varphi_{L}(p_{L}^*,I_{L}^*)+\frac{bI_{L}^{*}}{5}, \\
\tilde{n}_{F}^*=&\frac{4}{5}+\varphi_{F}(p_{L}^*,I_{L}^*)-\frac{bI_{L}^{*}}{5}
\end{align*}
In both equations, intuitively, the first term, $\frac{1}{5}$, $\frac{4}{5}$, represents the subscription from the common pool, if there had been no attrition to an outside option.  The second and third terms represent the impacts of the attrition as also the additions from the exclusive customer bases.  In the special case that $b=0$, i.e., when the demand functions depend only on the access fees,  the third term is $0$ and the demand functions capture the impact of attrition and additions in the expression for the subscriptions. For $b > 0$, the second and the third term together become $k - p_L^* + \frac{b}{5}I_L^* $ in the expression for $\tilde{n}_{L}^{*}$, and
$k - p_F^*+\frac{4b}{5}I_{L}^{*}$ in that for $\tilde{n}_{F}^{*}$. Thus, higher overall spectrum acquisition increases  the subscription for both SPs even in these terms. The intuitions remain same for the
second equilibrium-type solution, as the subscriptions are merely swapped.


Finally, when $L_{0}\geq 4/b, \Delta=0$, there does not exist an ``interior'' equilibrium-type solution, that is,   in which $0 < n_L, n_F < 1.$ Future research includes determining (1) whether there exists corner  equilibrium-type solutions, or (2) generalization to the case that $\Delta \neq 0$.

\subsection{Numerical results}\label{sec: outside numerical results}
\label{numericaloutside}
We set $b=2$, $k=c=1$, $w=0.2$ and $s=2$ throughout. For $s=2$ the condition for existence of interior  equilibrium-type solutions is satisfied for all cases below. Also $I_L^* < 4/b$ in all cases below.

\begin{figure}[hbt]
\begin{center}
  \includegraphics[width=1.7in, height=1.7in]{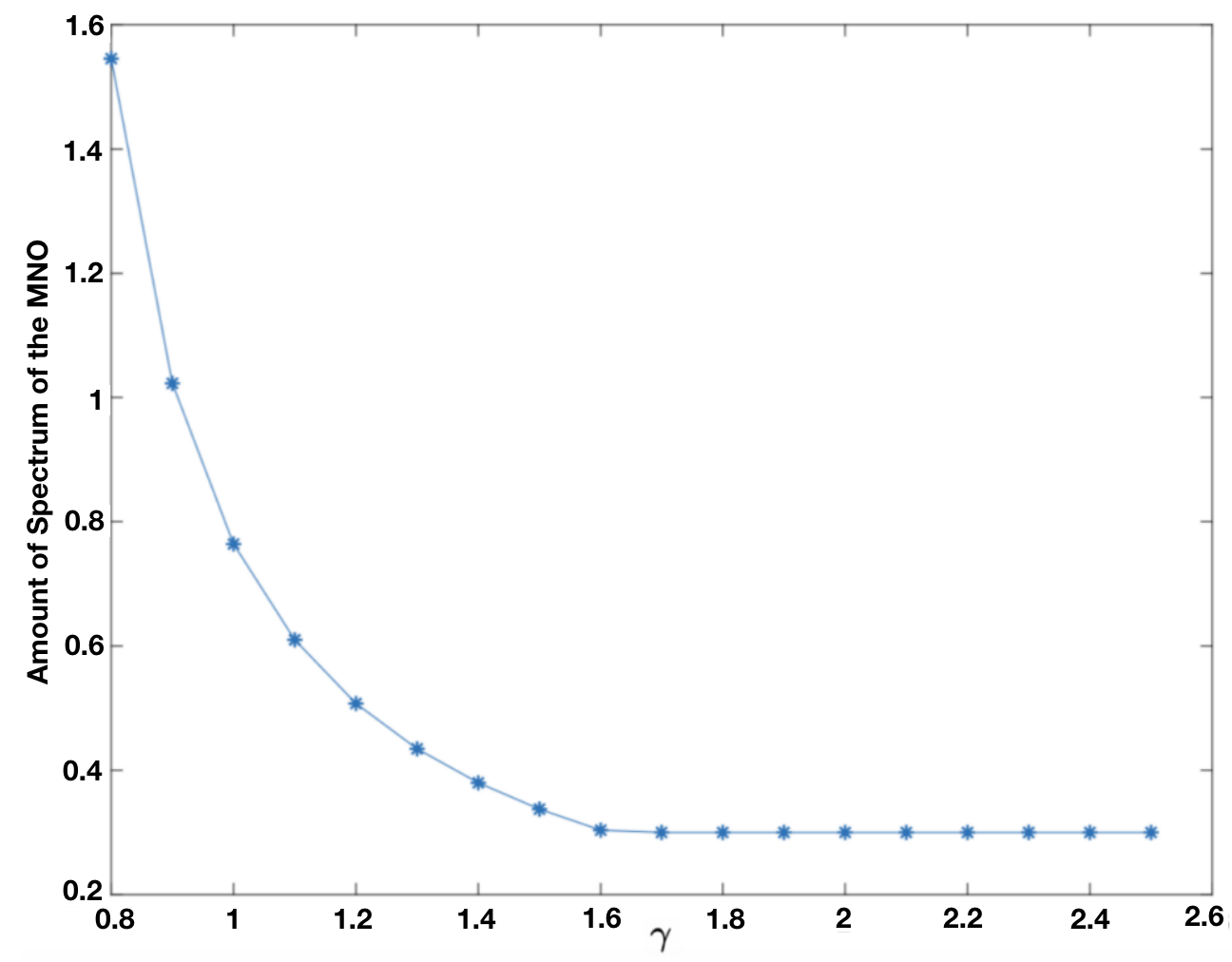}
  \includegraphics[width=1.7in, height=1.7in]{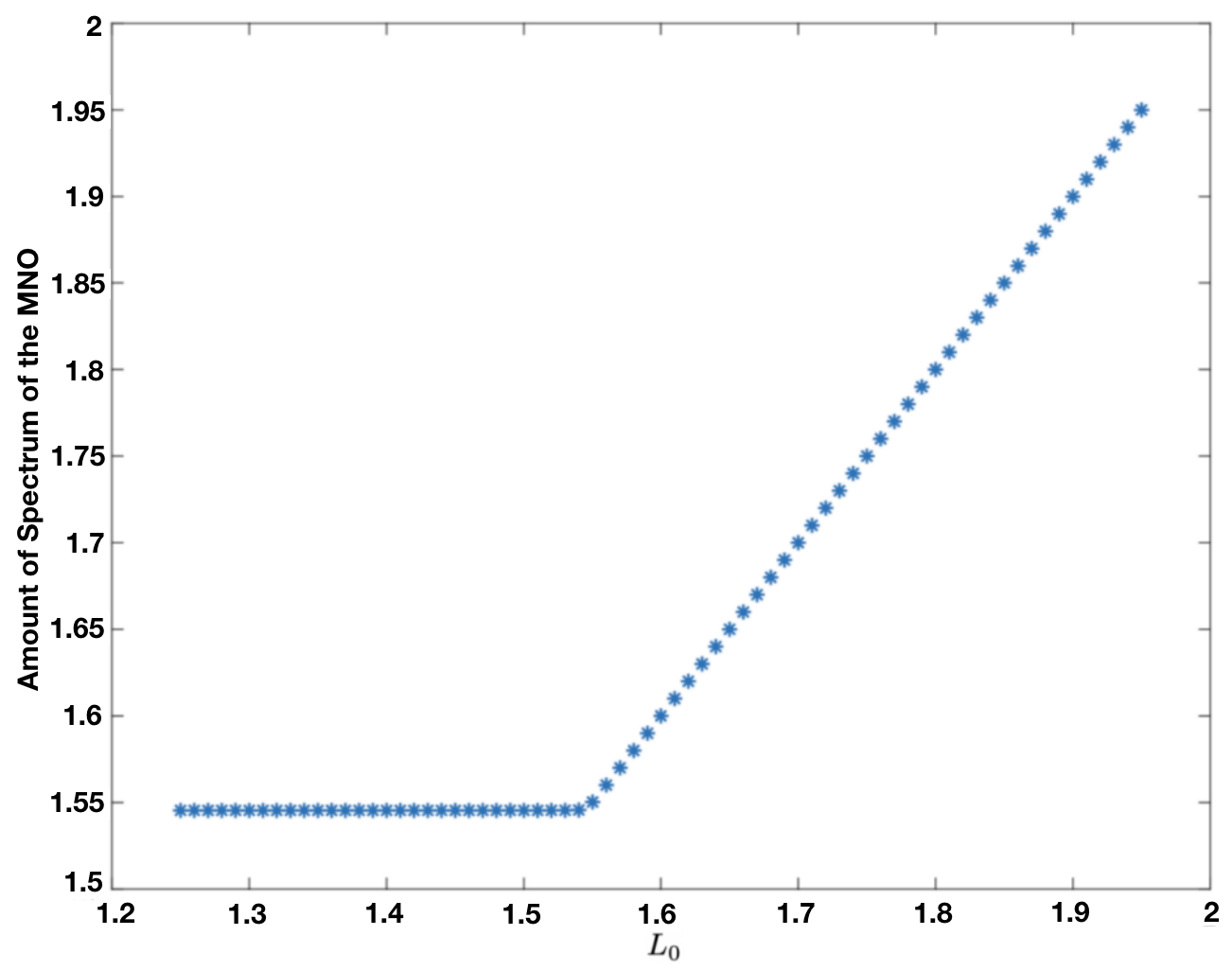}
\caption{The spectrum acquisition levels vs. $\gamma$ (left), $L_{0}$ (right)}
\label{figure:out-invest}	
\end{center}
\end{figure}

With $L_0 = 0.3$, Figure~\ref{figure:out-invest}~(left) shows that  the spectrum acquisitions for the two equilibrium-type solutions ($I_L^*$ is the same in both)  exceeds  $ L_0$ until $\gamma$ crosses a threshold, and subsequently remains at $L_0.$ Thus, SP$_L$ acquires more spectrum when it is cheaper to do so; otherwise settles at the minimum mandated amount. Now, with $\gamma=0.8$, Figure~\ref{figure:out-invest}~(right) shows that if $L_{0}$ is smaller than a threshold ($=1.54$), $I_{L}^{*}$ exceeds $L_0$ and equals the threshold value, and subsequently $I_{L}^{*}=L_{0}$. Thus, $I_{L}^{*}$ is initially constant and subsequently  increases linearly with $L_0.$

\begin{figure}[hbt]
\begin{center}
  \includegraphics[width=1.7in]{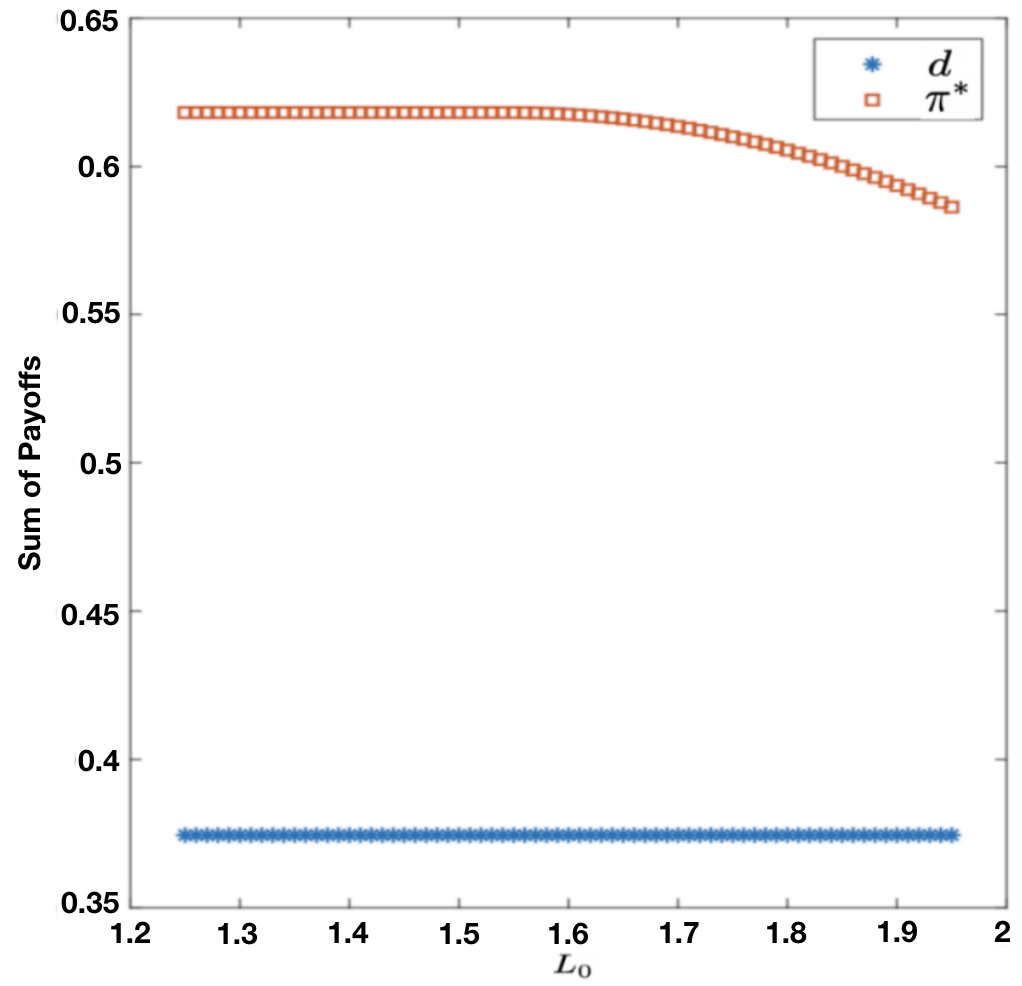}
    \includegraphics[width=1.7in]{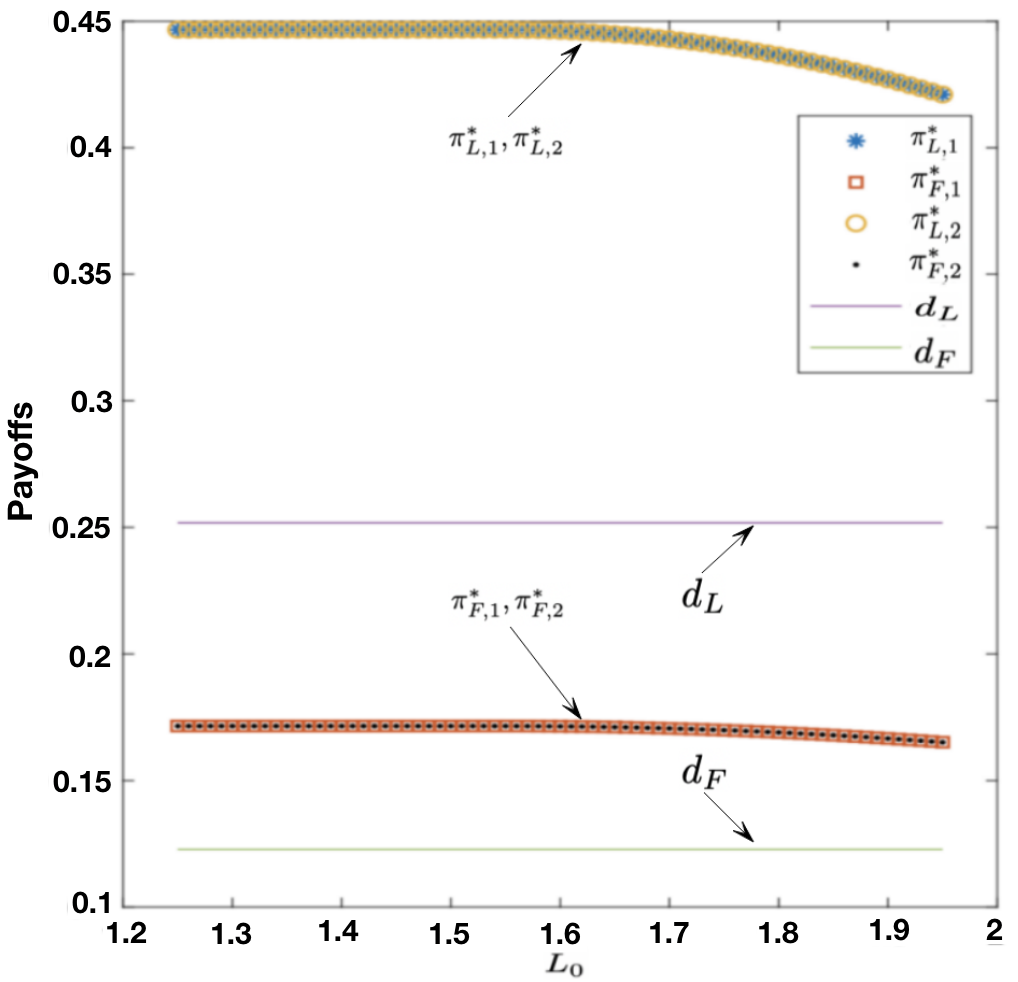}
\caption{Sum of equilibrium-type solution payoff, $\pi^*$, disagreement payoff, $d$ (left), payoffs of individual SPs (right) vs. $L_0$}
\label{fig-bargaining6}	
\end{center}
\end{figure}


Figure~\ref{fig-bargaining6} shows that the total payoff of the two SPs, as also their individual payoffs exceed the corresponding disagreement values, under both equilibrium-type solutions. As in the base case (eg, Figures~\ref{fig-bargaining1}, \ref{sepa_payoff1}, \ref{sepa_payoff2}) the total payoff and the individual payoffs decrease with increase in $L_0$, for the same reason as described in the paragraph after Corollary~\ref{cor: sum of payoffs independent of s*}. 
In the first equilibrium-type solution, SP$_{F}$ leases the entire spectrum SP$_L$ acquires, while in the second, SP$_L$ retains this entire spectrum. We observe $\pi_{L,1}^*<\pi_{F,1}^*$ and $\pi_{L,2}^*>\pi_{F,2}^*$. Thus, under Nash bargaining solution,   the SP that retains the entire spectrum  gets a higher share of the payoff.

\begin{figure}[hbt]
\begin{center}
  \includegraphics[width=2in]{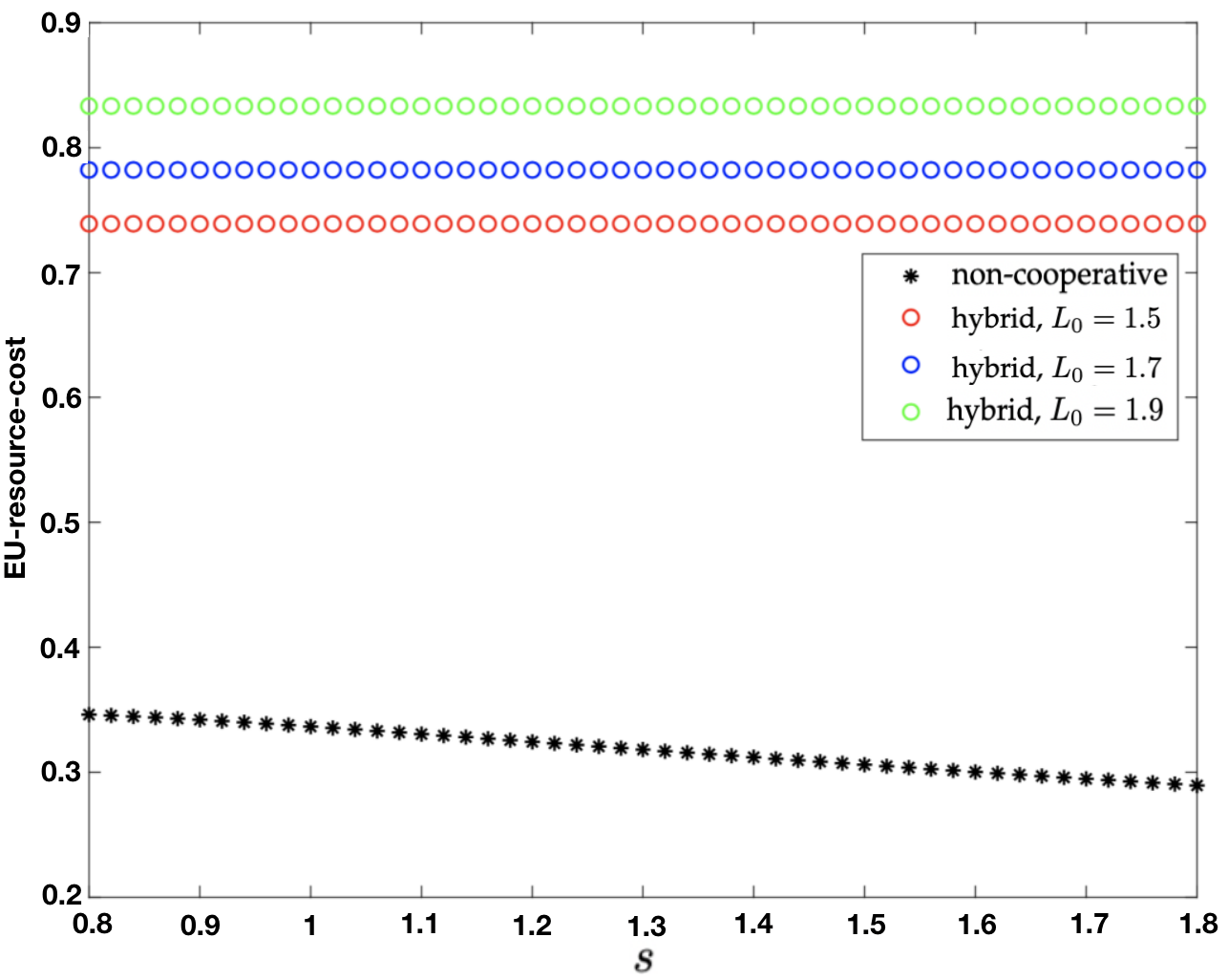}
\caption{The EU-resource-benefit vs $s$}
\label{fig-bargaining_a2}	
\end{center}
\end{figure}

From Theorem~\ref{thm: out-outcome-sectionA}, by simple calculation, the EU-resource-cost metric is $I_{L,1}^*/p_{F,1}^*$ in either SPNE. 
Figure~\ref{fig-bargaining_a2} shows that this metric is higher when the SPs jointly decide their spectrum acquisitions than when they decide separately (as in Part I~\cite{Part1}).  Naturally, for the joint decision case, this metric is constant with respect to $s$,  the reservation fee SP$_F$ pays to SP$_L$. Similar to the base case, this metric  increases with $L_0$.

\section{Generalization: Limited spectrum from the central regulator}\label{sec: Limited spectrum from the central regulator}
We now consider that SP$_L$ can lease at most $M$ spectrum units from the central regulator due to paucity at the latter's end, and
generalize the results in Sections~\ref{sec: cooperative base case}, \ref{sec: outside option}.
We naturally assume that $M \geq L_0$ (recall that $L_0$ is the minimum amount of spectrum that SP$_L$ is required to obtain).
\subsection{The base case}
Theorems~\ref{toprove}, \ref{thm: Un-outcome-1}, \ref{thm:existence}, \ref{thm: Un-outcome-4} in Section~\ref{sec: Un-outcome}, reveal that the SPNE strategies in the base case only depend on the lower bound $L_0$ of $I_L$. Thus, when $I_L$ is additionally required to be less than or equal to $M$,  SPNE strategies remain the same.

\subsection{EUs with outside option}
\label{sec: Limited spectrum from the central regulator-outside-option}
If an outside option exists in the system, then from Theorem~\ref{thm: out-outcome-sectionA}, the SPNE strategies  do not  only depend  on $L_0$. Then, Theorem~\ref{thm: out-outcome-sectionA} holds with
the constraint $L_0 \leq I_L$ in (1) replaced by $L_0 \leq I_L \leq M.$ We prove this in Appendix~\ref{APP: outside option,M}. The replacement is also intuitive.


 The numerical results in Section~\ref{sec: outside numerical results} reveal that without this upper bound $I_L^*$ is an interior point. Thus, if $M$ is relatively small, $I_L^*=M.$ If $M$ is large, then the values of  $I_L^*$ are as given by Theorem~\ref{thm: out-outcome-sectionA}, and as computed in Section~\ref{sec: outside numerical results}.

\section{Conclusions and Future Research}
This paper investigates the incentives of mobile network operators (MNOs) for acquiring
additional  spectrum to offer mobile virtual network operators (MVNOs) and thereby inviting competition
for a common pool of end users (EUs). We consider interactions between two service providers, a  MNO and an MVNO, when the EUs 1) must choose one of them 2) have the option to defect to  an outside option should the SP duo  offer unsatisfactory access fees or qualities of service. The 2 SPs jointly decide their spectrum acquisitions and the money flow between them, and separately decide the access fees for the EUs. We propose a multi-stage hybrid of cooperative bargaining and noncooperative games for modeling the interactions between the SPs, and  identify when the overall equilibrium solutions exist,  when it is unique and characterize the  equilibrium solutions when they exist.

Analytical and numerical results show that the payoffs of both SPs in this hybrid framework are higher than those in noncooperative framework (in Part~\uppercase\expandafter{\romannumeral1}~\cite{Part1}). In a market without outside option, EUs in this hybrid framework can attain higher or lower value of resource-cost tradeoff than that in noncooperative framework, while in a market with outside option, EUs typically attain a strictly higher value of this tradeoff than that in noncooperative framework.

Future research includes generalization  to accommodate: 1) arbitrary number of SPs 2) non-uniform distribution of EUs between the two SPs in the hotelling model,  3) distinct transaction costs $c_L$ and $c_F$, 4) potentially non-convex  spectrum reservation fee functions that the SP$_F$ pay the SP$_L$ and the SP$_L$ pay the regulator,   5) arbitrary transport cost $t_L, t_F$ functions of the spectrum acquired by the SPs, $I_L, I_F$. Considering 3 SPs as in Section~4 of Part~\uppercase\expandafter{\romannumeral1}~\cite{Part1} constitute a starting point towards 1). Possible starting points towards the others, as also that for moving from 3 SPs to arbitrary number of SPs,  have been provided in Section 5 of Part~\uppercase\expandafter{\romannumeral1}~\cite{Part1}.

\cleardoublepage

\appendices

\section{Proofs for Theorems in Section \ref{sec: Un-outcome}}\label{Appendix: Unequal1}



We prove Theorem~\ref{toprove} and Theorem~\ref{thm:existence} in two steps.

Let $|\Delta| < 1$. Consider $(I_L^*, I_F^*, p_L^*, p_F^*, n_L^*, n_F^*)$ that constitute the optimum solution of
\begin{equation}\label{equ: BM-bargainging optimization-3}
\begin{aligned}
\max_{I_{L}, I_{F}}\quad&u_{excess}\\
s.t.\quad&L_0\leq I_{L},\,\,0\leq I_{F}\leq I_{L}.
\end{aligned}
\end{equation}
  Here  $p_L^*, p_F^*, n_L^*, n_F^*$  are obtained from $I_L^*, I_F^*$  Theorem ~\ref{thm1Part1}:
 \begin{equation}\label{equ: Un-A-interior-prices}
\begin{aligned}
p_{L}^{*}=c+\frac{2}{3}-\frac{I_{F}}{3I_{L}}+\frac{\Delta}{3}, \ \
p_{F}^{*}=c+\frac{1}{3}+\frac{I_{F}}{3I_{L}}-\frac{\Delta}{3}.
\end{aligned}
\end{equation}
\begin{equation}\label{equ: Un-A-interior-subscriptions}
\begin{aligned}
n_{L}^{*}=p_L^*-c, \ \, n_{F}^{*}= p_F^*-c
\end{aligned}
\end{equation}

   In {\bf Step 1} we show that  any such $(I_L^*, I_F^*)$
must be of the form given in Theorem~\ref{toprove}. Next, note that an optimum solution of \eqref{equ: BM-bargainging optimization-3temp}, should it exist,  is also an optimum solution of \eqref{equ: BM-bargainging optimization-3}. Since equilibrium-type solutions constitute the optimum solutions of \eqref{equ: BM-bargainging optimization-3temp},  Theorem~\ref{toprove} follows.

In {\bf Step 2} we observe that given the $I_L^*, I_F^*$ of the possible equilibrium-type solutions mentioned in Theorem~\ref{toprove}, 1) $\tilde{s}^{*}, \theta^*$ of these can be  obtained from \eqref{equ: BM-bargaining-s} and  \eqref{equ: BM-bargaining-theta} respectively, and 2) $p_L^*, p_F^*, n_L^*, n_F^*$ of these can be obtained from Theorem~\ref{thm1Part1}.
 Accordingly,   Theorem~\ref{thm: Un-outcome-1} follows from Theorem~\ref{toprove}, as mentioned before Theorem~\ref{thm: Un-outcome-1}. The total payoff of the two SPs under each of the possible equilibrium-type solutions in Theorem~\ref{thm: Un-outcome-1} is the same, and is given in Corollary~\ref{cor: sum of payoffs independent of s*}. If any possible equilibrium-type solution listed in Theorem~\ref{toprove} is an equilibrium-type solution, then this total payoff must not exceed the sum of the disagreement payoffs.  Next, if this total payoff is not less than the disagreement payoffs, then $u_{excess} \geq 0$ under the possible equilibrium-type solutions listed in Theorem~\ref{toprove}. Thus, these solutions satisfy   the additional constraint in  \eqref{equ: BM-bargainging optimization-3temp} (beyond \eqref{equ: BM-bargainging optimization-3}), and therefore constitute its optimum solution too. Thus, these are equilibrium-type solutions . Theorem~\ref{thm:existence} follows.

\noindent{\bf Step 1.} \begin{proof} Consider $(I_L^*, I_F^*, p_L^*, p_F^*, n_L^*, n_F^*)$ that constitute the optimum solution of \eqref{equ: BM-bargainging optimization-3}.

Substituting \eqref{equ: Un-A-interior-prices} and \eqref{equ: Un-A-interior-subscriptions} into (\ref{equ: BM-L-payoff}) and (\ref{equ: BM-F-payoff}), we can get the payoffs of $\text{SP}_{F}, \text{SP}_{L}$, for some $\tilde{s}, \theta$  as:
\begin{equation}\label{equ: Fpayoff10}
\begin{aligned}
\pi_{F}^*=(\frac{1-\Delta}{3}+\frac{I_F^*}{3I_L^*})^2
-\tilde{s}(I_{F}^*)^2+\theta,
\end{aligned}
\end{equation}
  \begin{equation}\label{equ: Un-L-payoff1}
 \begin{aligned}
\pi_{L}^*=(\frac{\Delta+2}{3}-\frac{I_{F}^*}{3I_{L}^*})^{2}+\tilde{s}(I_{F}^*)^{2}-\gamma (I_{L}^*){2}-\theta.
\end{aligned}
\end{equation}

By Definition \ref{def: BM-u-excess}, substituting \eqref{equ: Fpayoff10} and \eqref{equ: Un-L-payoff1}  into (\ref{equ: BM-u-excess}), we can get $u_{excess}^*$,
\begin{equation}\label{equ: Un-bargainging optimization-3-sectionA}
\begin{aligned}
u_{excess}^*=&(\frac{\Delta+2}{3}-\frac{I_{F}^*}{3I_{L}^*})^{2}-\gamma (I_{L}^*){2}-d_{L}\\
+&(\frac{1-\Delta}{3}+\frac{I_{F}^*}{3I_{L}^*})^{2}-d_{F}.	
\end{aligned}
\end{equation}
Denote $t^*=I_{F}^*/I_{L}^*$, \eqref{equ: Un-bargainging optimization-3-sectionA} is equivalent to
\begin{equation}\label{equ: Un-uexcess-D}
\begin{aligned}
u_{excess}^*=&(\frac{\Delta+2-t^*}{3})^{2}-\gamma (I_{L}^*)^{2}\\
+&(\frac{1 -\Delta+t^*}{3})^{2}-d_{L}-d_{F}.
\end{aligned}
\end{equation}

Now we prove that $I_{L}^{*}=L_{0}$ by contradiction. Suppose $I_{L}^{*}>L_{0}$, then take $\hat{I}_{L}=L_{0}$ and $\hat{I}_{F}=I_{F}^{*}\frac{L_{0}}{I_{L}^{*}}$. Thus $t^{*}=I_{F}^{*}/I_{L}^{*}=\hat{I}_{F}/\hat{I}_{L}$. Since $t^{*}$ is constant and $\hat{I}_{L} < I_{L}^{*}$ , then $u_{excess}$ is higher with $\hat{I}_{F}$ and $\hat{I}_{L}$ than with $I_{F}^{*}$ and $I_{L}^{*}$. This contradicts the optimality of  $I_{F}^{*}$ and $I_{L}^{*}$. Therefore, $I_{L}^{*}=L_{0}$.

Take the second derivative of $u_{excess}$ with respect to $I_{F}$, $\frac{d^{2}u_{excess}}{dI_{F}^{2}}=\frac{4}{9I_{L}^{2}}>0$, then $u_{excess}$ is convex with respect to $I_{F}$, and the maximum of $u_{excess}$ must be obtained at the boundaries of $I_{F}$.

Then, we obtain the optimal solution $I_{F}^{*}$. Note $0\leq I_{F}\leq I_{L}$. Substitute $I_{F}^{*}=0$ and $I_{F}^{*}=I_{L}^{*}=L_{0}$ into (\ref{equ: Un-uexcess-D}), we have
\begin{align*}
u_{excess}(0,L_{0})-u_{excess}(L_{0},L_{0})=\frac{4}{9}\Delta.\\
\end{align*}
Therefore
\begin{equation*}
\begin{aligned}
&u_{excess}(0,L_{0})>u_{excess}(L_{0},L_{0})\quad\text{if}\quad \Delta>0\\
&u_{excess}(0,L_{0})=u_{excess}(L_{0},L_{0})\quad\text{if}\quad \Delta=0\\
&u_{excess}(0,L_{0})<u_{excess}(L_{0},L_{0})\quad\text{if}\quad \Delta<0
\end{aligned}
\end{equation*}
\begin{equation*}
\begin{aligned}
\Rightarrow&\left\{\begin{aligned}&(I_{F}^{*}, I_{L}^{*})=(0, L_{0})\quad\text{if}\quad 0<\Delta\leq 1\\
&(I_{F}^{*}, I_{L}^{*})=(0\,\,\text{or}\,\,L_{0},L_{0})\quad\text{if}\quad \Delta=0\\
&(I_{F}^{*}, I_{L}^{*})=(L_{0}, L_{0})\quad\text{if}\quad -1<\Delta<0\end{aligned}\right..
\end{aligned}
\end{equation*}


\end{proof}

\noindent{\bf Proof of Theorem~\ref{thm: Un-outcome-4}.}
\begin{proof}
Once $I_L^*, I_F^*$ are determined, $\tilde{s}^{*}$ is obtained by (\ref{equ: BM-bargaining-s}) and $\theta^*$ is obtained by \eqref{equ: BM-bargaining-theta}. We obtain $I_L^*, I_F^*, p_L^*, p_F^*$  in two steps: $\Delta\leq-1$ ({\bf Step 1}),  $\Delta\geq1$ ({\bf Step 2}).

\noindent{\bf Step 1: $\Delta\leq-1$ }
Suppose the reservation fee is $s$ in the sequential framework with $s>\gamma$.
 From Theorem~2~(3) in Part~\uppercase\expandafter{\romannumeral1}~\cite{Part1},  $n_L^*=0,n_F^*=1,$ \begin{equation}\label{equ: Un-A-corner-prices}
\begin{aligned}
&p_{L}^{*}=p_{F}^{*}+ \Delta -1\\
&p_{F}^{*}\in[c+1,c-\Delta-1].
\end{aligned}
\end{equation}
These also constitute the SPNE, together with,
\begin{align}\label{equ: corner investment level}
I_{L}'=I_{F}'=\frac{1}{\sqrt{2s}},	
\end{align}
 that provides the disagreement payoffs, $d_L, d_F$.
From (1), (2) in Part~\uppercase\expandafter{\romannumeral1}~\cite{Part1} and \eqref{equ: corner investment level}, $d_{L}+d_{F}=p_{F}^{*}-c-\frac{\gamma}{2s}.$

Again, from \eqref{equ: BM-L-payoff} and \eqref{equ: BM-F-payoff}, under equilibrium-type solution,
the payoffs of the SPs are
\begin{align}
\pi_{F}=&p_{F}^{*}-c-\tilde{s}^*(I_{F}^*)^{2}+\theta^*\label{equ: Fpayoff10-corner1},\\
\pi_{L}=&\tilde{s}^*I_{F}^{2}-\gamma (I_{L}^*)^{2}-\theta^*\label{equ: Un-L-payoff1-corner1}.
\end{align}

By Definition \ref{def: BM-u-excess}, substituting \eqref{equ: Fpayoff10-corner1} and \eqref{equ: Un-L-payoff1-corner1}   into (\ref{equ: BM-u-excess}):
\begin{equation}\label{equ: Un-uexcess1-corner1}
\begin{aligned}
u_{excess}^*=&p_{F}^*-c-\gamma (I_{L}^*)^{2}-d_{L}-d_{F}.
\end{aligned}
\end{equation}
Note that $u_{excess}^*$ is independent of $I_{F}$, then $I_{F}^{*}$ can be any number between $[0, I_{L}^{*}]$.
Therefore,  $I_{L}^{*}$ is a solution of the following optimization problem,
  \begin{equation}
\begin{aligned}
\max_{I_{L}, I_{F}}\quad&u_{excess}=p_{F}^*-c-\gamma I_{L}^{2}-d_{L}-d_{F}\\
s.t\quad&L_0\leq I_{L}\\
&u_{excess}\geq0
\end{aligned}
\end{equation}
From \eqref{equ: Un-A-corner-prices}, $p_{F}^{*}$ is independent of $I_{L}$, so the objective function is a decreasing function of $I_{L}$. Thus, $I_{L}^{*}=L_{0}$. Since $d_{L}+d_{F}=p_{F}^{*}-c-\frac{\gamma}{2s}$, then $u_{excess}^*\geq0$ is equivalent to $L_{0}\leq \frac{1}{\sqrt{2s}}.$ The result follows.

\noindent{\bf Step 2: $\Delta\geq1$:}
We first consider the corner SPNE for ($p_L, p_F, n_L, n_F)$ in  Theorem~2~(1) in Part~\uppercase\expandafter{\romannumeral1}~\cite{Part1}: $n_L^*=1,n_F^*=0$, and
\begin{equation}\label{equ: Un-A-corner-prices1}
\begin{aligned}
p_{F}^{*}=&p_{L}^{*}+v^{F}-v^{L}\\
p_{L}^{*}\in&[c+1,c+v^{L}-v^{F}].
\end{aligned}
\end{equation}
Along with $I_{L}'=\delta$, $I_{F}'=0$, these also constitute the SPNE that provide the disagreement payoffs.
Therefore, from (1) in Part~\uppercase\expandafter{\romannumeral1}~\cite{Part1},
$d_F=0$ and $d_L=p_{L}^*-c-\gamma \delta^2.$
From \eqref{equ: BM-L-payoff}, \eqref{equ: BM-F-payoff}, under an equilibrium-type solution,
\begin{equation}\label{equ: Fpayoff10-corner2}
\begin{aligned}
&\pi_{F}^*=\tilde{s}^*(I_{F}^*)^{2}+\theta^*\\
&\pi_{L}^*=p_{L}^{*}-c+\tilde{s}^*(I_{F}^*)^{2}-\gamma (I_{L}^*)^{2}-\theta^*,
\end{aligned}
\end{equation}
then substituting \eqref{equ: Fpayoff10-corner2}  into (\ref{equ: BM-u-excess}), we can get $u_{excess}^*$:
\begin{equation}\label{equ: Un-uexcess1-corner2}
\begin{aligned}
u_{excess}^*=&p_{L}^*-c-\gamma (I_{L}^*)^{2}-d_{L}-d_{F}.
\end{aligned}
\end{equation}
Note that $u_{excess}^*$ is independent of $I_{F}$, then $I_{F}^{*}$ can be any number between $[0, I_{L}^{*}]$.
Therefore, the optimum $I_{L}^{*}$ is a solution of the following optimization problem,
  \begin{equation}
\begin{aligned}
\max_{I_{L}, I_{F}}\quad&u_{excess}=p_{L}^*-c-\gamma I_{L}^{2}-d_{L}-d_{F}\\
s.t\quad&L_0\leq I_{L}\quad u_{excess}\geq0
\end{aligned}
\end{equation}
From \eqref{equ: Un-A-corner-prices1}, $p_{L}^{*}$ is independent of $I_{L}$, so the objective function is a decreasing function of $I_{L}$, then $I_{L}^{*}=L_{0}.$ Note that  $d_{L}+d_{F}=p_{L}^{*}-c-\gamma\delta^2$, then $u_{excess}^*\geq0$ is equivalent to $L_{0}\leq\delta.$

Next, we consider $\Delta=1$ and the interior SPNE in   Theorem~2~(2) in Part~\uppercase\expandafter{\romannumeral1}~\cite{Part1}, i.e., $0<n_F^*, n_L^*<1$. By similar analysis in Theorem~\ref{toprove}, we have $I_{L}^{*}=L_{0}$ and $I_{F}^*=0$. Therefore from \eqref{equ: Un-A-interior-prices} and \eqref{equ: Un-A-interior-subscriptions}, $p_{L}^{*}=c+1$, $p_{F}^*=c$, $n_{L}^*=1$, and $n_F^*=0$, which is contradicted to $0<n_F^*, n_L^*<1.$ Thus no equilibrium-type solution exists in this case.

\end{proof}

\noindent{\bf Proof of Theorem~\ref{thm: constant m*}.}
\begin{proof}
We calculate $m^*$ in $5$ cases: $-1<\Delta<0$, $\Delta=0$, $0<\Delta<1$, $\Delta\leq-1$ and $\Delta\geq1$.

\noindent{\bf Case 1.} When $-1<\Delta<0$. Note that $I_L^*-I_F^*=0$, then $m^*=I_F^*/p_F^*=L_0/(c+2/3-\Delta/3)$.

\noindent{\bf Case 2.} When $\Delta=0$. If $I_F^*=0$, then  then $m^*=(I_L^*-I_F^*)/p_L^*=L_0/(c+2/3)$. If $I_F^*=I_L^*$, then   $m^*=I_F^*/p_F^*=L_0/(c+2/3)$.

\noindent{\bf Case 3.} When $0<\Delta<1$. Note that $I_F^*=0$, then $m^*=(I_L^*-I_F^*)/p_L^*=L_0/(c+2/3+\Delta/3)$.
\end{proof}

\section{Proof of Theorems in Section~\ref{analysis}}\label{APP: outside option}
Once $I_L^*, I_F^*$ are determined, $\tilde{s}^{*}$ is obtained by (\ref{equ: BM-bargaining-s}) and $\theta^*$ is obtained by \eqref{equ: BM-bargaining-theta}. We therefore focus on obtaining $(I_L^*, I_F^*, p_L^*, p_F^*, n_L^*, n_F^*)$ corresponding to the equilibrium-type solutions.
Let $\Delta = 0$. Consider $(I_L^*, I_F^*, p_L^*, p_F^*, n_L^*, n_F^*)$ that constitute the optimum solution of \eqref{equ: BM-bargainging optimization-3} (with only the expressions for $u_{excess}$ differing from Appendix~\ref{Appendix: Unequal1}).
Per  Theorem 7 (3), (4), \cite{Part1}, for an interior SPNE, $I_L^* < 4/b$, and :\begin{equation}\label{equ: outsideoption-subscriptions}
\begin{aligned}
\tilde{n}_{L}^{*}=&\frac{I_{L}^{*}-I_{F}^{*}}{I_{L}^{*}}+p_{F}^{*}-2p_{L}^{*}+k+bI_{L}^{*}-bI_{F}^{*}\\
\tilde{n}_{F}^{*}=&\frac{I_{F}^{*}}{I_{L}^{*}}+p_{L}^{*}-2p_{F}^{*}+k+bI_{F}^{*},
\end{aligned}
\end{equation}
\begin{equation}\label{equ: price}
\begin{aligned}
p_{L}^{*}=&\frac{1}{15}+\frac{2c}{3}+\frac{k}{3}+\frac{t_{F}}{5}-\frac{b}{5}I_{F}+\frac{4b}{15}I_{L},\\
p_{F}^{*}=&\frac{1}{15}+\frac{2c}{3}+\frac{k}{3}+\frac{t_{L}}{5}+\frac{b}{15}I_{L}+\frac{b}{5}I_{F}.
\end{aligned}
\end{equation}

First, we show that  any such $(I_L^*, I_F^*, p_F^*, n_L^*, n_F^*)$
must be of the form given in Theorem~\ref{thm: out-outcome-sectionA}. Next, note that an optimum solution of \eqref{equ: BM-bargainging optimization-3temp}, should it exist,  is also an optimum solution of \eqref{equ: BM-bargainging optimization-3}. Since equilibrium-type solutions constitute the optimum solutions of \eqref{equ: BM-bargainging optimization-3temp},  Theorem~\ref{thm: out-outcome-sectionA} follows.


In fact, substituting \eqref{equ: Out-demand} and \eqref{demandfunction} in Part~\uppercase\expandafter{\romannumeral1}~\cite{Part1} into \eqref{equ: BM-L-payoff} and \eqref{equ: BM-F-payoff}, 
\begin{equation}\label{equ:stage 3}
\begin{aligned}
\pi_{F}=&\alpha(t_{L}+k+p_{L}-2p_{F}+bI_{F})(p_{F}-c)-\tilde{s}I_{F}^{2}+\theta\\
\pi_{L}=&\alpha(t_{F}+k+p_{F}-2p_{L}+bI_{L}-bI_{F})(p_{L}-c)\\
+&\tilde{s}I_{F}^{2}-\gamma I_{L}^{2}-\theta.
\end{aligned}
\end{equation}

\begin{lemma}
\label{theorem:g_bargain}
In any solution of \eqref{equ: BM-bargainging optimization-3}, 
$I_{F}^{*}=I_{L}^{*}$ or $I_F^* = 0$.
\end{lemma}
\begin{proof}
 By substituting (\ref{equ:stage 3}) into $u_{excess}=\pi_{L}-d_{L}+\pi_{F}-d_{F}$, and using $t_{L}=I_{F}/I_{L}$, $t_{F}=1-t_{L}$,
\begin{align*}
u_{excess}=&4\alpha f^{2}(I_{L})I_{F}^{2}-4\alpha f^{2}(I_{L})I_{L}I_{F}+2\alpha g^{2}(I_{L})\\
+&2\alpha(f(I_{L})I_{L}+g(I_{L}))^{2}-\gamma I_{L}^{2}-d_{F}-d_{L}.
\end{align*}


Next $\frac{d^2 u_{excess}}{d I^2_F}=8\alpha f^2(I_L)>0.$

Thus, $u_{excess}$ is convex wrt $I_F$, and the maximum of $u_{excess}$ is obtained at the boundary of $I_F$:  
\begin{align*}
&u_{excess}|_{I_{F}=I_{L}}=u_{excess}|_{I_{F}=0}\\
=&2\alpha g^{2}(I_{L})+2\alpha(f(I_{L})I_{L}+g(I_{L}))^{2}
-\gamma I_{L}^{2}-d_{F}-d_{L}.
\end{align*}
Thus $I_{F}^{*}=I_{L}$ or $I_{F}^{*}=0$.
\end{proof}

Also, for any solution of \eqref{equ: BM-bargainging optimization-3},  $I_{L}^{*}$ is given by
\begin{equation}\label{equ: proof-outside-I_L}
\begin{aligned}
\max_{I_{L}}\quad&2\alpha g^{2}(I_{L})+2\alpha(f(I_{L})I_{L}+g(I_{L}))^{2}-\gamma I_{L}^{2}\\
s.t\quad& L_0\leq I_{L}.\,\,
\end{aligned}
\end{equation}

Substituting $I_{F}^{*}=I_{L}^{*}$ and $I_{F}^{*}=0$ into \eqref{equ: outsideoption-subscriptions} and \eqref{equ: price}, combining with Lemma~\ref{theorem:g_bargain} and \eqref{equ: proof-outside-I_L}, it follows that any solution $(I_L^*, I_F^*, p_F^*, n_L^*, n_F^*)$ of \eqref{equ: BM-bargainging optimization-3}
must be of the form given in Theorem~\ref{thm: out-outcome-sectionA}. Thus, Theorem~\ref{thm: out-outcome-sectionA} follows.


From (4) in Part~\uppercase\expandafter{\romannumeral1}~\cite{Part1}, $x_{0}^*=t_{F}^*+p_{F}^{*}-p_L^*$,  substituting \eqref{equ: price}, $t_{F}^{*}=(I_L^*-I_F^*)/I_L^*$ into $x_0^*$, then we have $0<x_0^*<1$ if and only if $I_{L}^*<4/b$.
The total payoff of the two SPs under each of the possible interior equilibrium-type solutions listed in Theorem~\ref{thm: out-outcome-sectionA} is the same, and is given in Theorem~\ref{thm:existence_outside}. If any possible equilibrium-type solution listed in Theorem~\ref{thm: out-outcome-sectionA}  is an equilibrium-type solution, then this total payoff must not exceed the sum of the disagreement payoffs. Thus, the necessity in  Theorem~\ref{thm:existence_outside} follows.   Next, if $I_L^* < 4/b$, the $p_L^*, p_F^*$ in Theorem~\ref{thm:existence_outside} constitute an interior Nash equilibrium in Stage $2$ of the sequential hybrid game. If the total payoff of the possible equilibrium-type solutions in Theorem~\ref{thm:existence_outside} is not less than the disagreement payoffs, then $u_{excess} \geq 0$ under them. Thus, these solutions satisfy   the additional constraint in  \eqref{equ: BM-bargainging optimization-3temp} (beyond \eqref{equ: BM-bargainging optimization-3}), and therefore constitute its optimum solution too. 
Thus, the sufficiency in Theorem~\ref{thm:existence_outside} follows.

\qed

\section{Proof of results in Section~\ref{sec: Limited spectrum from the central regulator-outside-option}.}\label{APP: outside option,M}

In Section~\ref{sec: Limited spectrum from the central regulator-outside-option}, we had claimed the following Theorem:

\begin{theorem}\label{thm: out-outcome-sectionA, M}
Let $\Delta = 0$. Either there is no interior equilibrium-type solution, or there are two  interior equilibrium-type solutions. They are:
\begin{itemize}
\item[(1)] $I_{L,1}^{*}$ is a solution of
\begin{equation*}
\footnotesize
\begin{aligned}
\footnotesize
\max_{I_{L}}\,\,& 2\alpha g^{2}(I_{L})+2\alpha(f(I_{L})I_{L}+g(I_{L}))^{2}-\gamma I_{L}^{2}\\
s.t\,\,&L_0\leq I_{L}\leq M\,\,
\end{aligned}
\end{equation*}
$I_{F,1}^{*}=I_{L,1}^{*}$, $\tilde{s}^{*}$ is obtained by \eqref{equ: BM-bargaining-s}, and $\theta^*=0$.
\item[(2)]
$p_{L,1}^{*}=\frac{1}{15}+\frac{2c}{3}+\frac{k}{3}+\frac{bI_{L}^{*}}{15},\,
p_{F,1}^{*}=\frac{4}{15}+\frac{2c}{3}+\frac{k}{3}+\frac{4bI_{L}^{*}}{15}$.
\item[(3)] $\tilde{n}_{L,1}^{*}=\frac{2}{15}+\frac{2k}{3}+\frac{2bI_{L}^{*}}{15}-\frac{2c}{3},\, \tilde{n}_{F,1}^{*}=\frac{8}{15}+\frac{2k}{3}-\frac{2c}{3}+\frac{8bI_{L}^{*}}{15}$.
\end{itemize}
and
\begin{itemize}
\item[(1)] $I_{L,2}^{*} = I_{L,1}^*, I_{F,2}^{*}=0$, $\tilde{s}^{*}$ has no significance, and $\theta^*$ is obtained by \eqref{equ: BM-bargaining-theta}.
\item[(2)] $p_{L,2}^{*}=p_{F,1}^{*}$, $p_{F,2}^{*}=p_{L,1}^{*}$.
\item[(3)] $\tilde{n}_{L,2}^{*}=\tilde{n}_{F,1}^{*}$, $\tilde{n}_{F,2}^{*}=\tilde{n}_{L,1}^{*}$.
\end{itemize}
\end{theorem}

The proof of this Theorem is identical  to that for Theorem~\ref{thm: out-outcome-sectionA} in Appendix~\ref{APP: outside option}, with the following modification: the optimization problem \eqref{equ: proof-outside-I_L} becomes
\begin{equation*}
\begin{aligned}
\max_{I_{L}}\quad&2\alpha g^{2}(I_{L})+2\alpha(f(I_{L})I_{L}+g(I_{L}))^{2}-\gamma I_{L}^{2}\\
s.t\quad& L_0\leq I_{L}\leq M.
\end{aligned}
\end{equation*}
This is because $(I_L^*, I_F^*, p_L^*, p_F^*, n_L^*, n_F^*)$ constitute the optimum solution of
\begin{equation*}
\begin{aligned}
\max_{I_{L}, I_{F}}\quad&u_{excess}\\
s.t.\quad&L_0\leq I_{L}\leq M,\,\,0\leq I_{F}\leq I_{L}.
\end{aligned}
\end{equation*}
Theorem~\ref{thm: out-outcome-sectionA, M} now follows using arguments that are otherwise identical to that for the proof of Theorem~\ref{thm: out-outcome-sectionA}.

\end{document}